\documentclass[11pt,a4paper,reqno]{amsart}[tikz,border=10pt]
\usepackage[utf8]{inputenc}
\usepackage[top=1.5cm, bottom=1.5cm, left=2.0cm, right=2.0cm]{geometry}
\usepackage[utf8]{inputenc}
\usepackage{amsthm,amsfonts,amsmath,bookmark,appendix}
\usepackage{enumerate,graphicx,pdfpages,hyperref,dsfont,tabularx}
\usepackage{bbm}
\usepackage{graphicx}
\allowdisplaybreaks
\bigskip
\DeclareMathOperator{\Tr}{Tr}

\DeclareMathOperator{\conv}{conv}

\DeclareMathOperator{\eps}{\varepsilon}
\def\ra{\rangle}
\def\la{\langle}

\def\R {\mathbb{R}}
\def\Z {\mathbb{Z}}

\def\F {\mathcal{F}}
\def\T {\mathbb{T}}
\def\N {\mathbb{N}}

\def \D {\mathcal{D}}
\def\h{\mathcal{H}}

\def\R {\mathbb{R}}

\def\Fc{\mathcal{F}^{\text{can}}}
\def\Fcr{\mathcal{F}^{\text{can}}_\rho}
\def\fc{f^{\text{can}}}

\newtheorem{maintheorem}{Theorem}
\newtheorem{theorem}{Theorem}[section]
\newtheorem{prop}[theorem]{Proposition}
\newtheorem{cor}[theorem]{Corollary}
\newtheorem{lemma}[theorem]{Lemma}
\newtheorem*{defi}{Definition}

\theoremstyle{remark}
\newtheorem{remark}[theorem]{Remark}
\numberwithin{equation}{section}

\title[The Bogoliubov-Bose-Hubbard model]{The Bogoliubov-Bose-Hubbard model: existence of minimizers and absence of quantum phase transition}

\author[N. Mokrza\'nski]{Norbert Mokrza\'nski}
\address{Department of Mathematical Methods in Physics, Faculty of Physics, University of Warsaw,  Pasteura 5, 02-093 Warszawa, Poland}
\email{norbert.mokrzanski@fuw.edu.pl} 

\author[M. Napi\'orkowski]{Marcin Napi\'orkowski}
\address{Department of Mathematical Methods in Physics, Faculty of Physics, University of Warsaw,  Pasteura 5, 02-093 Warszawa, Poland}
\email{marcin.napiorkowski@fuw.edu.pl}

\begin{document}

\begin{abstract}
We consider a variational approach to the Bose-Hubbard model based on Bogoliubov theory. We introduce the grand canonical and canonical free energy functionals for which we prove the existence of minimizers. By analyzing their structure we show the existence of a thermally driven phase transition by showing that the system is superfluid  at sufficiently  low  temperatures and insulating at high temperatures. In particular, we show that this model does not exhibit a quantum phase transition.
\end{abstract}

\maketitle

\tableofcontents

\section{Introduction}

The Bose-Hubbard Model is a simple yet effective lattice model of locally interacting bosons. The model was introduced in \cite{GerKnol} as a bosonic version of the already established fermionic Hubbard Model. It became popular after it was realized in \cite{FWGF} that the model exhibits a phase transition, nowadays called the \textit{superfluid-Mott insulator} phase transition. Since then it has become one of the most important models in quantum many-body physics. 

Loosely speaking, the \textit{superfluid phase} occurs when the strength of interaction between the particles is sufficiently small. It is  related to the Bose-Einstein condensation phenomenon (macroscopic occupation of a single quantum state). One of possible characterizations of this phase is that the $U(1)$ symmetry generated by the particle number operator is spontaneously broken in the thermodynamic limit. The other -  \textit{Mott insulator} - phase is expected when the repulsion between the particles is sufficiently large -- in this situation it is energetically favorable for particles not to move between the lattice sites. It can be characterized by the existence of an energy gap in the spectrum. We refer to, e.g., \cite{Sachdev} for more details.

Proving the existence of the phase transition in many-body quantum systems is usually a very hard task. This remains true for the Bose--Hubbard model for which a proof of the superfluid-Mott insulator phase transition remains an open problem (see \cite{ALSSY,FerFrohUe,StaPuszWoj} for some rigorous results).

In this paper, we analyze the properties of the Bose-Hubbard model by introducing a variational perspective on the problem. In the case of continuous bosonic systems such approach  was first introduced by Critchley and Solomon in \cite{CritSolo} and was later rather successfully  expanded by Napiórkowski, Reuvers and Solovej in \cite{NapReuSol1, NapReuSol2, NapReuSol3} (see also \cite{FouNapReuSol-19,MokPal-24}). A similar approach was also considered by Bach, Lieb and Solovej in the context of the fermionic Hubbard model \cite{BachLiebSol}. The main idea behind this approach is based on Bogoliubov theory, first presented in Bogoliubov's famous paper \cite{Bogo}, where the Hamiltonian of the system is replaced by an "effective" one which is quadratic in creation and annihilation operators and therefore can be solved exactly \cite{NamNapSol-16}. We know that ground and Gibbs states of such Hamiltonians are quasi-free (or coherent) states. Here, we reverse the idea, and retain the full Hamiltonian while only varying over Gaussian states (which include the aforementioned classes of states - see Appendix \ref{app:derivation} for relevant definitions and a derivation).

This approach leads to considering the following expression, further referred to as the \textit{Bogoliubov-Bose-Hubbard (BBH) functional}:
\begin{equation}
\label{BBH}
\begin{split}
\F(\gamma, \alpha, \rho_0) &= \int_{\T^3}(\eps(p) - \mu)\gamma(p) dp - \mu \rho_0 - TS(\gamma, \alpha) \\&+ \frac{U}{2}\left(\int_{\T^3} \alpha(p)dp\right)^2 + U\left(\int_{\T^3}\gamma(p)dp\right)^2 \\&+ U\rho_0\int_{\T^3} \alpha(p) dp + 2U\rho_0\int_{\T^3} \gamma(p) dp + \frac{U}{2}\rho_0^2.
\end{split}
\end{equation}
Here $\T^3 = [-\pi,\pi]^3$ is the $3$-dimensional torus, $\eps(p)$ is the dispersion relation given by
$$\eps(p) = 2\sum_{j=1}^3 (1- \cos p_j) = 4\sum_{j=1}^3 \sin^2 \frac{p_j}{2}, \; \qquad p = (p_1,p_2,p_3) \in \T^3,$$
the parameter $U>0$ is the strength of the on-site repulsion between particles and $S$ is the entropy density functional given by
\begin{align*}
S(\gamma,\alpha) &= \int_{\T^3} s(\gamma(p),\alpha(p))dp =\int_{\T^3}\left[\left(\beta(p)+\frac{1}{2}\right)\ln\left(\beta(p)+\frac{1}{2}\right)-\left(\beta(p)-\frac{1}{2}\right)\ln\left(\beta(p)-\frac{1}{2}\right)\right]dp,
\end{align*}
where
\begin{equation}
\label{beta_def}
\beta(p):=\sqrt{\left(\frac{1}{2}+\gamma(p)\right)^{2}-\alpha(p)^{2}}.
\end{equation}
In every integral $dp$ denotes normalized Haar measure on torus $\T^3$, where normalization means that the measure of the whole space is equal to one.

The functional is defined on the domain $\mathcal{D}$ given by
\begin{equation}
\label{domain}
\mathcal{D}=\{(\gamma,\alpha,\rho_{0}) \colon \gamma \in L^1(\T^3),\gamma(p)\geq 0, \alpha(p)^{2}\leq\gamma(p)(1+\gamma(p)), \rho_{0}\geq 0\}.
\end{equation}
It will be convenient to also introduce the notation
\begin{equation}
\label{domain_slice}
\D' = \{(\gamma,\alpha) \colon \gamma \in L^1(\T^3),\gamma(p)\geq 0, \alpha(p)^{2}\leq\gamma(p)(1+\gamma(p))\},
\end{equation}
that is $\D'$ is a layer of the domain $\D$ containing the functions $\gamma$ and $\alpha$.

The physical interpretation of the variables $(\gamma,\alpha,\rho_{0})$ is as follows. The non-negative number $\rho_0$ is the density of particles in the zero-momentum state and its strictly positive value will be interpreted as the presence of Bose--Einstein condensation (BEC). The function $\gamma(p)$ describes the (thermal) distribution of particles in momentum space. Finally, the function $\alpha (p)$ describes the so-called {\it pairing} or {\it anomalous density} in the system.

With this in mind, different phases of the system are characterized by different structures of the variational minimizers. In particular, a non-zero $\rho_0$ (interpreted as the existence of BEC) corresponds to Off Diagonal Long-Range Order (ODLRO) in the one-body density matrix. Within our variational framework, the superfluid phase is identified by a non-vanishing anomalous density $\alpha(p)\not \equiv 0$ in the minimizer. This nonzero pairing amplitude is a manifestation of ODLRO and reflects the aforementioned broken global $U(1)$ gauge symmetry associated with particle-number conservation. A different perspective that supports this notion is provided by the Landau criterion of superfluidity: as shown by Bogoliubov \cite{Bogo}, the existence of a superfluid response requires the presence of pairing correlations, which give rise to a gapless, linear excitation spectrum. Finally, the insulator phase is characterized by $\alpha(p)\equiv 0$. Note that we use the term \textit{insulator phase} rather than \textit{Mott insulator} since the variational model considered here does not capture the full many-body structure of a true Mott-insulating state of the (approximate) form  $|\psi \rangle = \prod_{x \in \Lambda} \frac{1}{\sqrt{n_0!}}(a^*_x)^{n_0} |0 \rangle$ for some $n_0 \in \N$ (see Appendix \ref{app:derivation} for the relevant definitions concerning this formula). While it is believed that the zero-temperature Bose–Hubbard model exhibits only two quantum phases, the situation at positive temperatures becomes more complicated (see e.g. \cite{CapProSvi} for some physical overview and numerics). Therefore the insulating phase that we obtain should rather be viewed as a generic insulating (non-superfluid) state rather than the Mott-insulator.

Rather recently, our model, although considered only for zero temperature $T=0$ and formulated in a slightly different manner, appeared in the physics literature in \cite{Cirac} where numerical investigations showed that at zero temperature there is no quantum phase transition (see also \cite{Stoof}). To be more precise, it means that for fixed $\mu>0$ and $T=0$ there is no phase transition as one varies $U$. We confirm these findings rigorously, by showing that at zero temperature and for $\mu>0$ the system is always superfluid. From a physical point of view, the chemical potential $\mu$ determines the average particle density: negative values correspond to an almost empty lattice, while positive values make it energetically favorable to add particles and hence lead to a finite filling. In our variational description, a positive chemical potential relative to the single-particle band minimum (which in our case is zero) implies that the lowest mode is macroscopically occupied, marking the onset of Bose–Einstein condensation and superfluidity.
Furthermore, we show that the system exhibits a (thermally driven) phase transition by proving that at sufficiently low temperatures the system remains superfluid while for sufficiently high temperatures it is an insulator (for fixed $U>0$ and $\mu>0$). The latter remains true for $\mu \leq 0$ and all $T$ and $U$. We stress that the established phase transition occurs when the temperature is changed. It would be interesting to prove or disprove the existence of a phase transition in $U$ with both $T\neq 0$ and $\mu$ fixed. This remains an open problem.

% We can also show that the system is in the superfluid phase  at sufficiently low temperatures (for fixed $U>0$ and $\mu>0$). On the other hand, we can also show that 

% Our analysis implies also that for $\mu \leq 0$ the minimizer is a Mott insulator. 

% However, the existence of a phase transition at positive temperatures and $\mu>0$ remains an open question.

The functional can be also considered in the {\it canonical} formulation, that is for fixed total density $\rho:=\rho_0+\int\gamma$. This formulation imposes additional constraints on the variables and the proofs of the results in the grand canonical setting require some adjustments. However, the main conclusions regarding the existence of a phase transition remain the same. 

%We prove that in the canonical setting the system described by the functional also has a phase transition. Interestingly, the character of the phase transition is different than in the grand canonical setting. We will explain this in the next section where we will also state precisely all the results.

\section{Main results}

\subsection{Grand canonical approach.} Let us first consider the BBH functional in the grand canonical setting, i.e. the functional \eqref{BBH} on the domain $\mathcal{D}$ given in \eqref{domain}. We have the following fundamental well-posedness result.

\begin{maintheorem}
\label{main_existence}
For any $U > 0$, $\mu \in \R$ and $T \ge 0$ there exists a minimizer $(\gamma^{\min}, \alpha^{\min}, \rho_0^{\min})\in\mathcal{D}$ of the BBH functional \eqref{BBH}. If $\mu\leq 0$, then the minimizer is unique. Furthermore, for every minimizer the following equivalence holds: $\rho_0^{\min} = 0$ if and only if $\alpha^{\min}\equiv 0$.
\end{maintheorem}
Using this result and according to the discussion in the introduction, it is possible to define a state of the system.

\begin{defi}
\label{definition_phase}
Let $(\gamma^{\min}, \alpha^{\min}, \rho_0^{\min})$ be a minimizer of the BBH functional \eqref{BBH} for given values of the parameters $U$, $\mu$ and $T$. We will say that the system is
\begin{itemize}
    \item in the \textbf{superfluid phase} if $\alpha^{\min} \not \equiv  0$,
    \item in the \textbf{insulator phase} if $\alpha^{\min} \equiv 0$.
\end{itemize}
We will say that there occurs a phase transition with respect to a chosen parameter $U$, $\mu$ or $T$ if there exist values of such parameter where the system is in different phases (for other parameters being fixed). The lack of uniqueness of minimizers implies the possibility of the coexistence of phases.
\end{defi}

The second main result is about the phase diagram of the system. To state the result we define the function $J(T,\theta)$ for $T>0$ and $\theta \ge 0$ as
\begin{equation}
\label{J_func}
J(T,\theta) = \int_{\T^3}\left(e^{\frac{1}{T}\left(\eps(p) + \theta\right)}-1\right)^{-1} dp.  
\end{equation}
As the integrand is a pointwise increasing function with respect to $T$ (for fixed $\theta)$ and a decreasing function with respect to $\theta$ (for fixed $T$), it follows that the function $J(T,\theta)$ has the same monotinicity properties and, due to the Monotone Convergence Theorem, it is continuous in each variable (it is not necessarily jointly continuous). It is also easy to check that for fixed $\theta$ we have
\begin{equation}\label{eq:J(T)asymp}
\lim_{T \to 0^+}J(T,\theta) = 0, \quad \lim_{T \to +\infty}J(T,\theta) = +\infty
\end{equation}
and for fixed $T$
$$\lim_{\theta \to +\infty}J(T,\theta) = 0.$$
This function will be important in the analysis of structure of the minimizers.

We have the next main theorem of this paper.

\begin{maintheorem}
\label{main_structure}
Let $(\gamma^{\min}, \alpha^{\min}, \rho_0^{\min})$ be a minimizing triple of the BBH functional \eqref{BBH} for fixed values of $U$, $\mu$ and $T$. We have the following characterization of the phases:
\begin{enumerate}
    
    \item If $T = 0$ and $\mu \le 0$, then the minimizer is the vacuum: $\gamma ^{\min}\equiv 0$, $\alpha^{\min} \equiv 0$, $\rho_0^{\min} = 0$.
    
    \item If $T = 0$ and $\mu > 0$, then $\rho_0^{\min} > 0$, that is for zero temperature the minimizer exhibits condensation. Moreover $\gamma^{\min} > 0$ and $\alpha^{\min} < 0$.

    \item If $T>0$ and $\mu \le 0$ then the minimizer is unique and satisfies $\gamma^{\min} > 0$, $\alpha^{\min} \equiv 0$ and $\rho_0 = 0$.
    
    \item If $T>0$, $\mu > 0$ and $J(T,0) < \frac{\mu}{2U}$, then $\gamma^{\min} > 0$, $\alpha^{\min} < 0$ and $\rho_0^{\min} > 0$.

    \item If $T > 0$, $\mu > 0$ and $J(T,4\mu+2U) > \frac{\mu}{2U}$ then the minimizer is unique and satisfies $\gamma^{\min} > 0$, $\alpha^{\min} \equiv 0$ and $\rho_0^{\min} = 0$.
    \end{enumerate}
\end{maintheorem}
\noindent A qualitative phase diagram in the parameter space  $(T,\mu)$  is shown in Fig. \ref{fig:phase_diagram}. Notably, this diagram resembles the phase diagram of the continuous Bose gas presented in \cite[Fig. 1]{NapReuSol1}. However, in our case, we provide an explicit (albeit non-optimal) description of the curves that bound the unknown region of the phase diagram.

\begin{figure}[ht]
\centering
\includegraphics[width=0.6\textwidth]{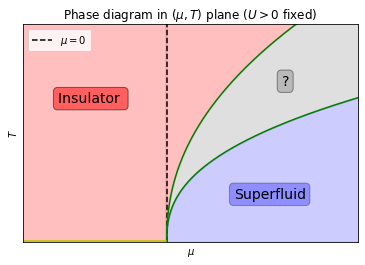}
\caption{Qualitative grand canonical phase diagrams of the model. The boundaries of the unknown area (marked as $"?"$) are given by the relations $J(T,0) = \frac{\mu}{2U}$ and $J(T,4\mu + 2U) = \frac{\mu}{2U}$.}
\label{fig:phase_diagram}
\end{figure}

\begin{remark}  Let us make the following remarks about the results stated above.
\begin{enumerate} 
\item Recall that we interpret the $\alpha(p)\equiv 0$ phase as the insulator phase. Thus, point $(2)$ of Theorem \ref{main_structure} implies the absence of a (quantum) phase transition at $T=0$ as discussed in the introduction.
\item Cases $(4)$ and $(5)$ show that the system exhibits, given fixed $\mu, U>0$, a phase transition as one varies the temperature $T$ from $0$ to $\infty$. This follows from the fact that  $J(T,0) > J(T,4\mu+2U)$ and the monotonicity given in \eqref{eq:J(T)asymp}.
\item Similarly, points $(1)$ and $(4)$ prove that there exists  a superfluid-insulator phase transition as one varies $\mu$. 
\item   Recall that $\rho_0$ is condensate density. Theorem \ref{main_existence} proves that BEC and superfluidity are equivalent in this model. The same was true for the Bose gas (see \cite[Theorem 2.5.]{NapReuSol1}).  
  \end{enumerate}
\end{remark} 
 
\noindent Proofs of Theorems \ref{main_existence} and \ref{main_structure} will be given in   Sections \ref{existence_section} and \ref{section_phase}. Right now we will briefly sketch the main ideas behind them. The general strategy follows \cite{NapReuSol1}. Concerning the first result, we start by proving that the functional \eqref{BBH} is bounded from below and convex in $(\gamma, \alpha)$ for fixed $\rho_0$. This allows to split the minimization procedure into two parts: first minimization with respect to $(\gamma, \alpha)$ for fixed $\rho_0$, then minimization of the result with respect to $\rho_0$. To overcome the occurring problem of non-reflexivity of the $L^1$ space, we also introduce the cut-off for the $\gamma$ function, that is we consider the domain on which $\gamma \le \kappa$ for some $\kappa > 0$. It will be a technical tool that will allow to use standard methods of the calculus of variations. This approach will give us a sequence of minimizers of the restricted problem -- we will then show that it is possible to extract a subsequence that is convergent (in a proper sense) to some $(\gamma, \alpha, \rho_0)$ and that this triple is a desired minimizer of the full problem. Our main tool at this stage will be analysis of the Euler-Lagrange equations and inequalities that are satisfied by the restricted minimizers. In order to obtain the final result, we will separate cases concerning zero and non-zero temperature, as we will use some different methods in each of them. In the non-zero temperature case we will further separate the problem into two sub-cases whether the minimizing $\rho_0$, appearing in the double minimization method, is positive or not. This will cover all the possibilities and eventually lead to a proof of Theorem \ref{main_existence}. All the details are presented in the Section \ref{existence_section}.

In order to prove Theorem \ref{main_structure} we will further analyse the Euler-Lagrange inequalities that arise from the existence of the unrestricted minimizers. In Section \ref{section_phase} we will use them to investigate the properties of the minimizers, which will then allow us to directly describe the relation between the values of the system's parameters and occurrence of different phases.

As already mentioned, the main steps of reasoning stem from \cite{NapReuSol1}. There are, however, some differences concerning certain aspects of the proofs. The first difference is that the integrals in \eqref{BBH} are defined over the compact set $\T^3$. For some arguments (mainly concerning the existence of minimizers) this is an advantage. For other parts of the analysis (e.g. concerning the phase diagram) this required modifications and refinements of the proofs.  Another difference is that the dispersion relation $\eps(p)$ is not spherically symmetric. One of the key properties of the system considered in \cite{NapReuSol1} was spherical symmetry of the minimizers (including restricted ones), which followed from the spherical symmetry of the dispersion relation of the model considered therein. This property was essential in the proof of existence of minimizers, therefore the method presented therein cannot be directly adapted to our situation. This issue led to a new strategy of the proof which didn't require such property. The main idea of the new method is to consider (in the double minimization procedure) the values of $\rho_0$ which are the minimizers of the restricted problems $\rho_{0,\kappa}$. This allows to use the variational derivative condition in this variable. A similar idea (for the continuous Bose gas) was used by Oldenburg in \cite{Old} to simplify the proof of existence of minimizers at zero temperature. Here this approach was adapted to the BBH functional, also at positive temperatures, which allowed not only to deduce existence of the minimizers, but also derive some of their properties used to prove the existence (or nonexistence) of the phase transition in certain parameter regimes. In particular our method allowed us to derive much more precise bounds on the critical region of the phase diagram than in \cite{NapReuSol1}.

\subsection{Canonical approach.} 
 
Let us now discuss the canonical setting. In the canonical formulation of the problem we are interested in finding minimizers of the BBH functional \eqref{BBH} under the condition that total density of particles in the system $\rho=\rho_0+\int\gamma$ is fixed. More precisely, we consider the canonical version of the functional \eqref{BBH} defined as
\begin{equation}
\label{BBHcan}
\begin{split}
\Fc(\gamma, \alpha, \rho_0) &= \int_{\T^3}\eps(p)\gamma(p) dp - TS(\gamma, \alpha) + \frac{U}{2}\rho^2
\\&+ \frac{U}{2}\left[\left(\int_{\T^3} \alpha(p)dp\right)^2 + \left(\int_{\T^3}\gamma(p)dp\right)^2\right]
\\&+ U\rho_0\left(\int_{\T^3} \alpha(p) dp + \int_{\T^3} \gamma(p) dp\right).
\end{split}
\end{equation}
Notice that this functional has the same form as \eqref{BBH} with chemical potential $\mu = 0$. We investigate existence of minimizers of the problem
\begin{equation}
\label{can_min}
\inf_{\D(\rho)} \Fc(\gamma, \alpha, \rho_0),   
\end{equation}
where, for fixed $\rho \ge 0$, we denoted
$$\D(\rho) = \left\{(\gamma, \alpha, \rho_0) \in \D \colon \int_{\T^3} \gamma(p) dp + \rho_0 = \rho\right\}.$$
Since for $\rho = 0$ this problem is trivial we will further assume that $\rho > 0$. 

As in this formulation $\rho_0$ is fully determined by $\rho$ and $\gamma$, we can reformulate the above minimization problem as follows: let
$$\D'(\rho) = \left\{ (\gamma, \alpha) \in \D' \colon \int_{\T^3}\gamma(p)dp \le \rho\right\}$$
and, for $(\gamma, \alpha) \in \D'(\rho)$, let
$$\Fcr(\gamma, \alpha) = \Fc\left(\gamma, \alpha, \rho - \int \gamma\right).$$
The minimization problem \eqref{can_min} is equivalent to finding a minimizer of
\begin{equation}
\label{can_min2}
\inf_{\D'(\rho)}\Fcr(\gamma,\alpha).  
\end{equation}
We shall prove the following result, which is a canonical setting analogue of Theorems \ref{main_existence} and  \ref{main_structure}.
\begin{maintheorem}
\label{can_existence}
For any $U > 0$, $\rho \ge 0$ and $T \ge 0$ there exist a minimizer of the problem \eqref{can_min} (or equivalently for the problem \eqref{can_min2}). For given $\rho >0$ denote $T_1$ as the temperature such that $J(T_1,0) = \rho$ and $T_2$ as the temperature such that $J(T_2, 2U\rho) = \rho$. If $(\gamma^{\min}, \alpha^{\min}, \rho_0^{\min}) \in \D(\rho)$ is a minimizing triple, then
\begin{enumerate}
    \item if $T < T_1$ (including $T=0$) we have $\rho_0^{\min} > 0$,
    \item if $T \ge T_2$ we have $\rho_0^{\min} = 0$ and the minimizer is unique.
\end{enumerate}
\end{maintheorem}

Thus, similarly as proven in Theorem \ref{main_structure}, the canonical BBH model exhibits a thermal phase transition and has no quantum phase transition.

\begin{remark}
We believe that in the canonical (resp. grand canonical) setup, the system has one critical temperature given by  the equation $J(T_c,0) = \rho$  (resp. $J(T_c,0) = \frac{\mu}{2U}$). 
\end{remark}
The proof of Theorem \ref{can_existence} will be given in Section \ref{canonical}. It is a modification of the proofs in the grand canonical approach and therefore only the main changes will be highlighted.

\section{Existence of minimizers}
\label{existence_section}

\subsection{Boundedness from below}
Let us introduce the  notation for the density of particles outside the condensate
$$\int_{\T^3}\gamma(p) dp = \rho_\gamma.$$
With this, we can rewrite the functional $\F$ as
\begin{equation}
\label{func2}
\begin{split}
\F(\gamma, \alpha, \rho_0) &= \int_{\T^3}\eps(p)\gamma(p) dp - \mu (\rho_\gamma + \rho_0) - TS(\gamma, \alpha) \\&+ \frac{U}{2}\left(\int_{\T^3} \alpha(p)dp\right)^2 + U\rho_0\int_{\T^3} [\gamma(p) + \alpha(p)] dp \\& + U\rho_\gamma^2 + U\rho_0\rho_\gamma + \frac{U}{2}\rho_0^2.
\end{split}
\end{equation}
We start with a simple yet fundamental proposition.
\begin{prop}
\label{boundedness}
The BBH functional $\F$ is bounded from below.
\end{prop}

\begin{proof} The condition $\alpha^2 \le \gamma(\gamma+1)$ implies that for almost all $p \in \T^3$ we have
$$\gamma(p) + \alpha(p) \ge -\frac12.$$
This implies that
\begin{equation}
\label{constantC}
\begin{split}
&U\rho_0\int_{\T^3} [\gamma(p) + \alpha(p)] dp + U\rho_\gamma^2 +U\rho_0\rho_\gamma + \frac{U}{2}\rho_0^2 - \mu (\rho_\gamma + \rho_0)
\\&\ge -\frac{U}2 \rho_0 + \frac{U}{2}\rho_0^2 + \frac{U}{2}\rho_\gamma^2 - \mu(\rho_0 + \rho_\gamma)
\\& = \left[\frac{U}{2}\rho_0^2 - \left(\frac{U}{2} + \mu\right)\rho_0\right] + \left[\frac{U}{2}\rho_\gamma^2 - \mu\rho_\gamma\right]
\\&= \frac{U}{2}\left(\rho_0 - \frac{U + 2\mu}{U}\right)^2 + \left(\rho_\gamma - \frac{2\mu}{U}\right)^2 - \left(\frac{U}{8} + \frac{\mu}{2} + \frac{\mu^2}{U}\right).
\end{split}
\end{equation}
Now let us focus on kinetic and entropy terms of the functional \eqref{func2}:
$$\int_{\T^3}\eps(p)\gamma(p) dp - TS(\gamma, \alpha).$$
For $T=0$ this expression is positive, so it is trivially bounded below. For $T>0$, by computing the variational derivative of the above functional, we get
\begin{equation}
\label{C_T}
\int_{\T^3} \eps(p)\gamma(p)dp - TS(\gamma, \alpha)  \ge T\int_{\T^3} \ln\left(1 - e^{-\eps(p)/T}\right)dp
\end{equation}
with minimizer
$$\gamma_0(p) = (e^{\eps(p)/T} -1)^{-1}.$$
Note that the integral in \eqref{C_T} is finite as in the neighborhood of $p=0$ we have $\eps(p) = p^2 + o(p^2)$ and therefore $\ln (1 - e^{-\eps(p)/T}) \sim \ln (p^2/T)$, which is an integrable function on $\T^3$. Thus we can write
\begin{equation}
\label{kin_ent_est}
T\int_{\T^3} \ln\left(1 - e^{-\eps(p)/T}\right)dp =: -C_T > -\infty, 
\end{equation}
where $C_T$ is the absolute value of this integral. Since every remaining term in functional \eqref{func2} is positive, this ends the proof.
\end{proof}
From the above proof we have obtained an explicit lower bound on the functional.
\begin{cor}
\label{bounded_seq}
For any $(\gamma, \alpha, \rho_0) \in \D$ we have
$$\F(\gamma, \alpha, \rho_0) \ge \frac{U}{2}\left(\rho_0 - \frac{U + 2\mu}{U}\right)^2 + \left(\rho_\gamma - \frac{2\mu}{U}\right)^2 - \left(\frac{U}{8} + \frac{\mu}{2} + \frac{\mu^2}{U}\right) - C_T,$$
where $C_T$ was defined in \eqref{kin_ent_est}. In particular, for sequences $(\gamma_n, \alpha_n, \rho_{0,n})$ minimizing the BBH functional $\F$ the sequence $\rho_{0,n}$ is bounded and $\gamma_n$ is bounded in $L^1(\T^3)$. Additionally, since $\alpha_n \le \sqrt2 \max \{1, \gamma_n\}$, sequence $\alpha_n$ is also bounded in $L^1(\T^3).$
\end{cor}

\subsection{Convexity and double minimization}

The next part of the proof concerns the double minimization procedure. We will first analyze convexity of the functional $\F$.
\begin{prop}
\label{convexF}
The domain $\D$ is a convex set and the BBH functional $\F$ is jointly convex in $(\gamma, \alpha)$. If $T>0$ then for a fixed $\rho_0$ the functional $\F$ is strictly jointly convex in $(\gamma, \alpha)$.
\end{prop}

\begin{proof}
Convexity of $\D$ follows from direct computation. To prove convexity of $\F$ first check by direct computation that Hessian of entropy function $-s(\gamma,\alpha)$, regarded as function of $\gamma$ and $\alpha$, is positive definite, therefore the function $(-s)$ is strictly convex in $(\gamma, \alpha)$ (note that this term contributes only when $T > 0$). Now notice that we can rewrite the functional $\F(\gamma, \alpha, \rho_0)$ in the following way:
\begin{equation}
\label{conv_expr}
\begin{split}
\F(\gamma, \alpha, \rho_0) &=
\int_{\T^3}(\eps(p) - \mu)\gamma(p) dp - \mu\rho_0 - TS(\gamma, \alpha)
\\&+ \frac{U}{2}\left(\int_{\T^3} \alpha(p) dp + \rho_0\right)^2 + U\left(\int_{\T^3} \gamma(p) dp \right)^2 + 2U\rho_0 \int_{\T^3} \gamma(p) dp.
\end{split}
\end{equation}
This functional is jointly convex in $(\gamma, \alpha)$, due to the fact that every term in the above sum for fixed $\rho_0$ is convex in $(\gamma, \alpha)$. Indeed, the non-linear terms in $\gamma$ or $\alpha$ are the following:
\begin{itemize}
    \item entropy part, which we already mentioned, is strictly convex in $(\gamma, \alpha)$
    \item quadratic terms, which are jointly convex (even in three variables $(\gamma, \alpha, \rho_0)$) due to the fact that function $(x,y) \mapsto (x+y)^2$ is convex.
\end{itemize} 
For $T>0$ and fixed $\rho_0$ functional $\F$ is strictly convex in $(\gamma, \alpha)$ due to existence of strictly convex entropy term.
\end{proof}

Now we proceed by splitting the minimization problem into two parts. Define function $f$ as
\begin{equation}
\label{double_min}
f(\rho_0) = \inf_{\D'} \F(\gamma, \alpha, \rho_0),
\end{equation}
where $\D'$ was introduced in \eqref{domain_slice}. Then we have
$$\inf_{\D} \F(\gamma, \alpha, \rho_0) = \inf_{\rho_0} f(\rho_0).$$
We will investigate some properties of the function $f$.
\begin{lemma}
\label{f_cont}
Function $f$ is continuous on the interval $[0,+\infty)$.
\end{lemma}
\begin{proof}
First consider a "convexified" version of the BBH functional
\begin{equation}
\begin{split}
\F^{\conv}(\gamma, \alpha, \rho_0) &=
\int_{\T^3}(\eps(p) - \mu)\gamma(p) dp - \mu\rho_0 - TS(\gamma, \alpha)
\\&+ \frac{U}{2}\left(\int_{\T^3} \alpha(p) dp + \rho_0\right)^2 + U\left(\int_{\T^3} \gamma(p) dp + \rho_0\right)^2
\\& = \F(\gamma, \alpha, \rho_0) + U\rho_0^2.
\end{split}
\end{equation}
Using arguments as in the previous proposition we deduce $\F^{\conv}$ is jointly convex in three variables. Now define $f^{\conv}(\rho_0)$ as 
$$f^{\conv}(\rho_0) = \inf_{\D'} \F^{\conv}(\gamma, \alpha, \rho_0).$$
This function is convex, which in particular implies continuity on the open interval $(0,\infty)$.

Continuity of $f^{\conv}$ at the point $\rho_0 = 0$ requires some extra arguments. From convexity we have the inequality
$$f^{\conv}(0) \ge \lim_{\rho_0 \to 0+} f^{\conv}(\rho_0) =:g.$$
To prove the reverse inequality we will show that there exists a sequence $(\gamma_n, \alpha_n) \in \D'$ such that
\begin{equation}
\label{cont_claim}
\lim_{n \to \infty}\F^{\conv}(\gamma_n, \alpha_n, \rho_0=0) = g.  
\end{equation}
Every element of this sequence is bounded below by $f^{\conv}(0)$, so this will imply the desired inequality. In order to prove existence of such sequence note that for any $\rho_0$ we can find $\gamma$ and $\alpha$ such that $\F^{\conv}(\gamma, 
\alpha, \rho_0)$ is arbitrarily close to $f^{\conv}(\rho_0)$ (i.e. $\gamma$ and $\alpha$ are parts of the minimizing sequence). It follows that, using the diagonal method, we can find a sequence $(\gamma_n,\alpha_n,\rho_{0,n})$ with $\rho_{0,n} \to 0$ such that
\begin{equation*}
\lim_{n \to \infty} \F^{\conv}(\gamma_n,\alpha_n,\rho_{0,n}) = g.   
\end{equation*}
We will show that
\begin{equation}
\label{cont_claim2}
\lim_{n \to \infty} \left(\F^{\conv}(\gamma_n,\alpha_n,\rho_{0,n}) - \F^{\conv}(\gamma_n,\alpha_n,0)\right) = 0.   
\end{equation}
We have
\begin{equation}
\label{func_difference}
\begin{split}
&\F^{\conv}(\gamma_n,\alpha_n,\rho_{0,n}) - \F^{\conv}(\gamma_n,\alpha_n,0)
\\&= -\mu \rho_{0,n} + \frac{U}{2}\rho_{0,n}\left(2\int_{\T^3} \alpha_n(p) dp + \rho_{0,n}\right) + U\rho_{0,n}\left(2\int_{\T^3} \gamma_n(p) dp + \rho_{0,n}\right)   
\end{split}    
\end{equation}
By Corollary \ref{bounded_seq} both sequences $(\gamma_n)$ and $(\alpha_n)$ are uniformly bounded in $L^1(\T^3)$, hence passing to the limit $n \to \infty$ in \eqref{func_difference} proves \eqref{cont_claim2} from which the claim \eqref{cont_claim} follows directly.

% Next, using uniform continuity of $\F^{\conv}$ with respect to $\rho_0$ (to be more precise: uniform equicontinuity with respect to $\rho_0$ for every $(\gamma, \alpha)$) we conclude elements of the form $(\gamma_n, \alpha_n, 0)$ also converge to $g$. This ends the proof of continuity of $f^{\conv}$.

As a result we also obtain continuity of $f$ since
$$f(\rho_0) = f^{\conv}(\rho_0) - U\rho_0^2.$$
\end{proof}

Observe that using Corollary \ref{bounded_seq} we get
$$\lim_{\rho_0 \to \infty} f(\rho_0) = +\infty.$$
Combining it with the lemma above  gives us the following result.
\begin{cor}
\label{rho_0_min}
There exists a $\rho_0^{\min}$ such that
$$\inf_{\D} \F(\gamma, \alpha, \rho_0) = f(\rho_0^{\min}).$$
In particular $\rho_0^{\min}$ minimizes the function $f$.
\end{cor}
Moreover, using strict convexity of the BBH functional in $(\gamma, \alpha)$ (cf. Proposition \ref{convexF}), we can deduce yet another corollary.

\begin{cor}
\label{gamma_alpha_unique}
For $T>0$ and given $\rho_0$, from the strict convexity of the entropy term, if there exist $(\gamma^{\min},\alpha^{\min}) \in \D'$ such that
$$f(\rho_0) = \F(\gamma^{\min}, \alpha^{\min}, \rho_0),$$
then $\gamma^{\min},\alpha^{\min}$ are unique.
\end{cor}

\subsection{Restricted problem and the existence of restricted minimizers}
\label{restricted_subsection}
Now we will introduce the restricted problem. For $\kappa > 0$ define
\begin{equation*}
\D_\kappa=\{(\gamma,\alpha,\rho_{0}) \colon \gamma \in L^1(\T^3),\gamma(p)\geq 0, \alpha(p)^{2}\leq\gamma(p)(1+\gamma(p)), \gamma(p) \le \kappa, \rho_{0}\geq 0\}.
\end{equation*}
and
$$\D_\kappa' = \{(\gamma,\alpha) \colon \gamma \in L^1(\T^3),\gamma(p)\geq 0, \alpha(p)^{2}\leq\gamma(p)(1+\gamma(p)), \gamma(p) \le \kappa\}.$$
As mentioned in the Introduction, introducing a cut-off will allow us to overcome problems related to the fact that $L^1$ is not reflexive. We are interested in finding minimizers of the functional $\F$ over the restricted domain $\D_\kappa$ (we will further refer to it as the restricted problem). Just like in the unrestricted case, we can split the minimization problem into two parts. Let
\begin{equation}
\label{mimimization1}
f_\kappa(\rho_0) = \inf_{D'_\kappa} \F(\gamma, \alpha, \rho_0).  
\end{equation}
As the restricted domain $\D_\kappa$ is also a convex set, we obtain analogous results as in Corollary \ref{rho_0_min}.

\begin{cor}
\label{rho_0_min_kappa}
There exists a $\rho_{0,\kappa}$ such that
$$\inf_{\D_\kappa} \F(\gamma, \alpha, \rho_0) = f_\kappa(\rho_{0,\kappa}).$$
In particular, $\rho_{0,\kappa}$ minimizes the function $f_\kappa$. 
\end{cor}

We will prove that every value of $f_\kappa$ is in fact obtained as a value of the functional $\F$.
\begin{prop}
\label{restriced_existence}
For every $\rho_0$ there exist $(\gamma_\kappa, \alpha_\kappa) \in \D_\kappa'$ (dependent on $\rho_0$) such that
\begin{equation*}
f_\kappa(\rho_0) = \F(\gamma_\kappa, \alpha_\kappa, \rho_0). 
\end{equation*}
In particular, there exists a minimizer of the restricted problem, that is
$$\inf_{\D_\kappa}\F(\gamma,\alpha, \rho_0) = \F(\gamma_\kappa, \alpha_\kappa, \rho_{0,\kappa})$$
for some $(\gamma_\kappa, \alpha_\kappa, \rho_{0,\kappa}) \in \D_\kappa$.
\end{prop}
\begin{proof} \textit{Step 1.} The cutoff implies that $\|\gamma\|_\infty < \infty$ for every $\gamma$ in the restricted domain. Since the measure of the space $\T^3$ is finite, we deduce that for any $s \in [1,\infty]$ we have $$\|\gamma\|_s < \kappa < \infty.$$
Next recall that condition on $\gamma$ and $\alpha$ from the domain \eqref{domain} implies $|\alpha(p)| \le \sqrt2 \max \{1, \gamma\} \le \sqrt2 \kappa$, so we also get
$$\|\alpha\|_s < \sqrt2 \kappa < \infty.$$

\noindent \textit{Step 2.} Let $(\gamma_n, \alpha_n) \in \D'_\kappa$ be a minimizing sequence of \eqref{mimimization1}. Since the sequence is uniformly bounded in the reflexive space $L^2(\T^3)$ (or $L^s(\T^3)$ for any $1<s<\infty$), by Banach-Alaoglu theorem (see e.g \cite[Theorem 2.18]{LiebLoss}), up to subsequences, we have weak convergence
$$\gamma_n \rightharpoonup \tilde \gamma,$$
$$\alpha_n \rightharpoonup \tilde \alpha$$
for some $\tilde \gamma$ and $\tilde \alpha$.
By Mazur's lemma (see e.g. \cite[Theorem 2.13]{LiebLoss}) there exist sequences of convex combinations of $\gamma_n$ and $\alpha_n$ that converge strongly in $L^2$ and in $L^1$ (due to estimates from Step 1) to $\tilde \gamma$ and $\tilde \alpha$ respectively. By convexity of the BBH functional $\F$ in $(\gamma, \alpha)$, they are still minimizing sequences. Denoting those sequences by $\gamma_n$ and $\alpha_n$ and once again passing to subsequences we can additionally assume that we have pointwise convergence:
$$\gamma_n(p) \to \tilde \gamma(p)$$
$$\alpha_n(p) \to \tilde \alpha(p)$$
for almost all $p \in \T^3$.

\noindent \textit{Step 3.} We will show that $\tilde\gamma(p)$ and $\tilde\alpha(p)$ are desired minimizers of the restricted problem. First recall that just like in the proof of Proposition \ref{boundedness} we can estimate integrands of kinetic and entropy part of the functional for $T>0$
$$\eps(p)\gamma(p)dp - Ts(\gamma, \alpha)  \ge T\ln\left(1 - e^{-\eps(p)/T}\right).$$
This function is integrable, so we can use Fatou lemma to obtain
\begin{equation*}
\int_{\T^3} \eps(p)\tilde \gamma(p)dp - TS(\tilde\gamma, \tilde\alpha)  \le \liminf_{n \to \infty} \left(\int_{\T^3} \left(\eps(p)\gamma_n(p)dp - TS(\gamma_n, \alpha_n)\right) dp\right).
\end{equation*}
For $T=0$ this is even simpler, because the integrand is positive, so Fatou lemma is also applicable. Every other term in the functional converges directly from the fact that $\gamma_n$ and $\alpha_n$ converge in $L^1$.

Recalling that $(\gamma_n, \alpha_n)$ was a minimizing sequence we get
$$\F(\tilde \gamma, \tilde \alpha, \rho_0) \le \liminf_{n \to \infty} \F(\gamma_n, \alpha_n, \rho_0) = \inf_{(\gamma, \alpha) \in \D'_\kappa}\F(\gamma, \alpha, \rho_0),$$
which ends the proof.
\end{proof}

In the next proposition we will prove that the values of the functional $\F$ taken on the restricted minimizers approximate the infimum of $\F$ on the whole (i.e. unrestricted) domain $\mathcal{D}$.

\begin{prop}
\label{minimizing_seq}
If $(\gamma_\kappa, \alpha_\kappa, \rho_{0,\kappa})$ are minimizers of the restricted problems, then
$$\lim_{\kappa \to \infty} \F(\gamma_\kappa, \alpha_\kappa, \rho_{0,\kappa}) = \inf_{\D} \F(\gamma, \alpha, \rho_0).$$
Similarly, for a fixed $\rho_0^{\min}$ minimizing the BBH functional $\F$ (cf. Corollary \ref{rho_0_min}) and corresponding restricted minimizers $(\gamma_\kappa, \alpha_\kappa)$ have
$$\lim_{\kappa \to \infty} \F(\gamma_\kappa, \alpha_\kappa, \rho_0^{\min}) = \inf_{\D} \F(\gamma, \alpha, \rho_0).$$
\end{prop}

\begin{proof}
We will show only the first part as the second one is completely analogous. Notice that for any state $(\gamma, \alpha, \rho_0) \in \D$ and for any $\kappa$ we have $(\gamma\mathbbm{1}_{\gamma \le \kappa}, \alpha\mathbbm{1}_{\gamma \le \kappa}, \rho_0) \in \D_\kappa$ and by the Dominated Convergence Theorem (using estimate on kinetic and entropy part like in the proof of Proposition \ref{boundedness}) we have
$$\F(\gamma, \alpha, \rho_0) = \limsup_{\kappa \to \infty}\F(\gamma\mathbbm{1}_{\gamma \le \kappa}, \alpha\mathbbm{1}_{\gamma \le \kappa}, \rho_0) \ge \limsup_{\kappa \to \infty}\F(\gamma_\kappa, \alpha_\kappa, \rho_{0,\kappa}).$$
Taking infimum over $(\gamma, \alpha, \rho_0)$ we get
$$\inf_\D \F(\gamma, \alpha, \rho_0) \ge \limsup_{\kappa \to \infty}\F(\gamma_\kappa, \alpha_\kappa, \rho_{0,\kappa}).$$
From the definition of the infimum we also get the reverse inequality
$$\inf_\D \F(\gamma, \alpha, \rho_0) \le \liminf_{\kappa \to \infty}\F(\gamma_\kappa, \alpha_\kappa, \rho_{0,\kappa}),$$
which ends the proof.
\end{proof}

Recall the double minimization procedure \eqref{double_min} for the unrestricted problem. Our goal is to prove that if $\rho_{0,\kappa}$ are restricted minimizers then $\rho_{0,\kappa} \to \rho_0^{\min}$ for some minimizer $\rho_0^{\min}$ of the unrestricted problem. To this end we will prove a lemma that is a continuation of the previous proposition.
\begin{lemma}
If $(\gamma_\kappa, \alpha_\kappa, \rho_{0,\kappa})$ are minimizers of the restricted problems, then
$$\lim_{\kappa \to \infty}f(\rho_{0,\kappa}) =\inf_{\rho_0}f(\rho_0).$$
\end{lemma}
\begin{proof}
In the proof of the previous proposition we got equality
$$\inf_\D \F(\gamma, \alpha, \rho_0) = \lim_{\kappa \to \infty}\F(\gamma_\kappa, \alpha_\kappa, \rho_{0,\kappa}),$$
so rewriting it in terms of functions $f$ and $f_\kappa$ we get
$$\inf_{\rho_0}f(\rho_0) = \lim_{\kappa \to \infty} f_\kappa(\rho_{0,\kappa}) \ge \limsup_{\kappa \to \infty} f(\rho_{0,\kappa}),$$
where the last inequality comes from the fact that $f_\kappa$ is defined as an infimum over $\gamma$ and $\alpha$ over smaller set than the infimum in the definition of the function $f$. From the definition of the infimum we also get the inequality
$$\liminf_{\kappa \to \infty} f(\rho_{0,\kappa}) \ge \inf_{\rho_0}f(\rho_0).$$
This ends the proof.
\end{proof}
\begin{prop}
\label{rho_0_conv}
If $(\gamma_\kappa, \alpha_\kappa, \rho_{0,\kappa})$ are minimizers of the restricted problems, then, up to a subsequence,
$$\lim_{\kappa \to \infty}\rho_{0,\kappa} = \rho_0^{\min}$$
where $\rho_0^{\min}$ is a minimizer of function $f$ (cf. Corollary \ref{rho_0_min}).
\end{prop}
\begin{proof}
By Corollary \ref{bounded_seq} we already know $\rho_{0,\kappa}$ is bounded, so we can extract a convergent subsequence. For simplicity of notation we will denote the chosen convergent subsequence as $\rho_{0,\kappa}$ and assume
$$\rho_{0,\kappa} \to \tilde \rho_0$$
for some $\tilde \rho_0$. We will check that $\tilde \rho_0$ is a minimizer of $f$. Indeed, by the previous lemma and continuity of $f$ (cf. Lemma \ref{f_cont}) we have
$$\inf_{\rho_0} f(\rho_0) = \lim_{\kappa \to \infty} f(\rho_{0,\kappa}) = f(\tilde \rho_0).$$
so $\tilde \rho_0$ is a minimizer.
\end{proof}

\subsection{Variational derivatives}

In the next steps we will provide estimates on minimizers of the restricted problems. We will use variational derivatives of functional $\F$, denoted by $\frac{\partial \F}{\partial \gamma}$, $\frac{\partial \F}{\partial \alpha}$ and $\frac{\partial \F}{\partial \rho_0}$, defined as functions satisfying
$$\frac{d}{dt}\F(\gamma + t\phi, \alpha, \rho_0)\Big|_{t=0} = \int_{\T^3} \frac{\partial \F}{\partial \gamma}(p) \phi(p) dp$$
for a test function $\phi$ (analogously defined for $\alpha$ and $\rho_0$).
By direct computation we can check that
\begin{equation}
\label{variational_derivatives}
\begin{split}
\frac{\partial \F}{\partial \gamma}  &= \eps(p) - \mu + 2U\int_{\T^3} \gamma(q)dq + 2U\rho_0 - T \frac{\gamma + \frac12}{\beta}\ln\frac{\beta + \frac12}{\beta - \frac12},\\
\frac{\partial \F}{\partial \alpha} &= U\int_{\T^3} \alpha(q)dq + U\rho_0 + T\frac{\alpha}{\beta}\ln\frac{\beta + \frac12}{\beta - \frac12},\\
\frac{\partial \F}{\partial \rho_0} &= -\mu + U\int_{\T^3} \alpha(q)dq + 2U\int_{\T^3} \gamma(q)dq + U\rho_0,
\end{split}
\end{equation}
where $\beta$ was defined in \eqref{beta_def}. In the next sections we will separate the cases when $T>0$ and $T=0$ since we will be using some different methods in each of them.
\begin{remark}
In the following we will use the notational convention that
$$\frac{\partial \F}{\partial \gamma} \equiv \frac{\partial \F}{\partial \gamma}\bigg|_{(\gamma^{\min},\alpha^{\min},\rho_0^{\min})},$$
(and similarly for $\alpha$ and $\rho_0$ derivatives), i.e., the functional derivative is evaluated at the minimizers (possibly restricted, depending on the context).
\end{remark}

\subsection{Positive temperature case}
In every proposition in this subsection we assume $T>0$. First we will prove some bounds on the minimizing $\gamma_\kappa$ and $\alpha_\kappa$.

\begin{lemma}
\label{bound_below}
Let $\gamma_\kappa$ be the minimizer of restricted problem. Then there exist constant $c$ independent of $\kappa$ such that for sufficiently large $\kappa$ we have
$$\gamma_\kappa(p) \ge c$$
for almost all $p \in \T^3$.
\end{lemma}
\begin{proof}
We will estimate the $\frac{\partial \F}{\partial \gamma}(p)$ derivative for $T>0$ on the set $\{p \colon \gamma_\kappa(p) < \kappa\}$. Notice that, since $\beta_\kappa \ge \frac12$, we have
$$T \frac{\gamma_\kappa + \frac12}{\beta_\kappa}\ln\frac{\beta_\kappa + \frac12}{\beta_\kappa - \frac12} \ge T \ln \frac1{\beta_\kappa - \frac12}.$$
It follows that, on this set, we have
$$\frac{\partial \F}{\partial \gamma}(p) \le \eps(p) - \mu + 2U\int_{\T^3} \gamma_\kappa(q)dq + 2U\rho_{0,\kappa} - T \ln \frac1{\beta_\kappa - \frac12}.$$
Since we already know that $(\gamma_\kappa,\alpha_\kappa, \rho_{0,\kappa})$ is a minimizing sequence of the unrestricted functional (cf. Proposition \ref{minimizing_seq}), by Corollary \ref{bounded_seq} we can bound $\|\gamma_\kappa\|_1$ and $\rho_{0,\kappa}$ by constant $C$ dependent only on $U$ and $\mu$, hence 
$$\eps(p) - \mu + 2U\int_{\T^3} \gamma_\kappa(q)dq + 2U\rho_{0,\kappa} - T \ln \frac1{\beta_\kappa - \frac12} \le \eps(p) + C - T \ln \frac1{\beta_\kappa - \frac12}.$$
For minimizing $\gamma_\kappa$ the set $\{p \colon \frac{\partial \F}{\partial \gamma} < 0\} \cap \{p \colon \gamma_\kappa(p) < \kappa\}$ has to have zero measure (otherwise one could decrease the energy by increasing $\gamma$, which is possible on this set without violating the domain constraints). Note that the set of elements $p \in \T^3$ such that
$$\eps(p) + C - T \ln \frac1{\beta_\kappa - \frac12} < 0 \iff \beta_\kappa < e^{\frac{-\eps(p) - C}{T}} + \frac12$$
is a subset of the above set, so it also has to have zero measure. In particular, since $\beta_\kappa \le \gamma_\kappa + \frac12$ this inequality holds for elements $p$ such that
$$\gamma_\kappa(p) <  e^{\frac{-\eps(p) - C}{T}}.$$
This means that for almost every $p \in \{p \colon \gamma_\kappa(p) < \kappa\}$ for minimizing $\gamma_\kappa$ we have
$$\gamma_\kappa(p) \ge e^{\frac{-\eps(p) - C}{T}}.$$
Because $\eps(p)$ is bounded, it also implies that for minimizing $\gamma_\kappa$ we have
$$\gamma_\kappa(p) > c > 0$$
for almost every $p \in \{p \colon \gamma_\kappa(p) < \kappa\}$, where $c$ is independent of $\kappa$:
$$c = c(U, \mu, T).$$ 
For sufficiently large $\kappa$ we have $\kappa > c$, so, for such $\kappa$, we obtain
$$\gamma_\kappa(p) \ge c$$
for almost every $p \in \T^3$.
\end{proof}

% \begin{remark}
% \label{gamma_lower_bound}
% The above proof in fact shows that
% $$\gamma_\kappa(p) \ge e^{-\frac1{T}\left(- \mu + 2U\int_{\T^3} \gamma_\kappa(q)dq + 2U\rho_{0,\kappa}\right)}.$$
% Assuming the minimizer $(\gamma, \alpha, \rho_0)$ of the unrestricted problem exists, we can repeat the same argument to show that it satisfies
% $$\gamma(p) \ge e^{-\frac1{T}\left(- \mu + 2U\int_{\T^3} \gamma(q)dq + 2U\rho_{0}\right)}.$$
% This bound will be useful in the proof of the phase transition.
% \end{remark}

\begin{lemma}
\label{alpha_ineq}
For $T>0$ minimizers of the restricted functional fulfill the strict inequality
\begin{equation*}  
\alpha^2_\kappa(p) < \gamma_\kappa(p)(\gamma_\kappa(p)+1)
\end{equation*}
for almost all $p \in \T^3$.
\end{lemma}
\begin{proof}
First suppose that
$$\alpha_\kappa^2(p) \ge \frac12 \gamma_\kappa(p)(\gamma_\kappa(p)+1).$$
With this assumption we have
$$\left|\frac{\alpha_\kappa}{\beta_\kappa}\right| \ge \frac{\sqrt{\frac12\gamma_\kappa(p)(\gamma_\kappa(p)+1)}}{\sqrt{\gamma_\kappa(p)(\gamma_\kappa(p)+1) + \frac14 - \alpha_\kappa^2(p)}} \ge \frac{\sqrt{\frac12\gamma_\kappa(p)(\gamma_\kappa(p)+1)}}{\sqrt{\gamma_\kappa(p)(\gamma_\kappa(p)+1) + \frac14}} =: B(\gamma_\kappa(p)).$$
Furthermore, as $\|\gamma_\kappa\|_1$ and $\rho_{0,\kappa}$ are uniformly bounded by a constant independent of $\kappa$, one can use the bound $\|\alpha\|_1 \le c\|\gamma\|_1$ (see Corollary \ref{bounded_seq}) to estimate
$$\big| U\int_{\T^3} \alpha(q)dq + U\rho_0 \big| \le C$$
for some constant $C$ dependent only on $U$ and $\mu$. Using this for $p$ such that $\alpha_\kappa(p) > 0$ we can write
$$\frac{\partial \F}{\partial \alpha}(p) \ge - C + TB(\gamma_\kappa(p))\ln\frac1{\beta_\kappa - \frac12}$$
and for $p$ such that $\alpha(p) < 0$
$$\frac{\partial \F}{\partial \alpha}(p) \le  C - TB(\gamma_\kappa(p))\ln\frac1{\beta_\kappa - \frac12}.$$
From this we can see that if
$$\beta_\kappa \le \frac12 + e^{-C/TB(\gamma_\kappa)}$$
we can decrease the energy by varying $\alpha_\kappa$. Since $\alpha_\kappa$ is a minimizer, the set of $p$'s where such inequality holds must be of measure zero. We can rewrite the above inequality as
$$\gamma_\kappa(\gamma_\kappa+1) + \frac14 - \alpha_\kappa^2 \le \left(\frac12 + e^{-C/TB(\gamma_\kappa)}\right)^2$$
and
$$\alpha_\kappa^2 \ge \gamma_\kappa(\gamma_\kappa+1) + \frac14 - \left(\frac12 + e^{-C/TB(\gamma_\kappa)}\right)^2 = \gamma_\kappa(\gamma_\kappa+1) - e^{-C/TB(\gamma_\kappa)} - e^{-2C/TB(\gamma_\kappa)}.$$
It follows that on a set of full measure we need to have the strict inequality
$$\alpha_\kappa^2(p) < \gamma_\kappa(p)(\gamma_\kappa(p)+1).$$
Note that we already know $\gamma_\kappa(p) \ne 0$ on a set of full measure, so above inequality is well stated.
\end{proof}
As a corollary we obtain that the variational derivative of $\F$ with respect to $\alpha$ must be zero. For further reference, we collect all relations coming from variational derivatives.
\begin{cor}
\label{EL+}
Variational derivatives \eqref{variational_derivatives} evaluated on the minimizers of the restricted problem $(\gamma_\kappa, \alpha_\kappa, \rho_{0,\kappa})$ satisfy the following relations:
\begin{align*}
\frac{\partial \F}{\partial \gamma} &= \eps(p) - \mu + 2U\int_{\T^3} \gamma_\kappa(q)dq + 2U\rho_{0,\kappa} - T \frac{\gamma_\kappa + \frac12}{\beta_\kappa}\ln\frac{\beta_\kappa + \frac12}{\beta_\kappa - \frac12} = \left\{ 
\begin{array}{l}
     = 0 \text{ if } \gamma_\kappa(p) < \kappa \\
     \le 0\text{ if } \gamma_\kappa(p) = \kappa
\end{array}\right.,\\
\frac{\partial \F}{\partial \alpha} &= U\int_{\T^3} \alpha_\kappa(q)dq + U\rho_{0,\kappa} + T\frac{\alpha_\kappa}{\beta_\kappa}\ln\frac{\beta_\kappa + \frac12}{\beta_\kappa - \frac12} = 0,\\
\frac{\partial \F}{\partial \rho_0}&= -\mu + U\int_{\T^3} \alpha_\kappa(q)dq + 2U\int_{\T^3} \gamma_\kappa(q)dq + U\rho_{0,\kappa} =\left\{\begin{array}{l}
     = 0 \text{ if } \rho_{0,\kappa} > 0 \\
     \ge 0\text{ if } \rho_{0,\kappa} = 0
\end{array}\right..
\end{align*}
The first two relations also hold if $(\gamma_\kappa, \alpha_\kappa$) are minimizers of the restricted problem with fixed $\rho_0$.
\end{cor}

Using this corollary we obtain the our first result on the structure of the minimizers.
\begin{prop}
\label{rest_stucture}
Let $(\gamma_\kappa, \alpha_\kappa, \rho_{0,\kappa})$ be the minimizer of the restricted problem. Then
$$\alpha_\kappa \equiv 0 \Longleftrightarrow \rho_{0,\kappa} = 0.$$
If $\rho_{0,\kappa} > 0$, then $\alpha_\kappa$ is negative almost everywhere and satisfies strict norm inequality
$$\|\alpha_\kappa\|_1 < \rho_{0,\kappa}.$$
\end{prop}
\begin{proof}
From the Euler-Lagrange equation for $\alpha$ (second equation in Corollary \ref{EL+}) it follows that
\begin{equation}
\label{ELalpha}
T\frac{\alpha_\kappa}{\beta_\kappa}\ln\frac{\beta_\kappa + \frac12}{\beta_\kappa - \frac12} = - U\int_{\T^3} \alpha_\kappa(q)dq - U\rho_{0,\kappa} = const,  
\end{equation}
so $\alpha_\kappa$ has a fixed sign, because the expression
\begin{equation}
\label{expression}
T\frac{1}{\beta_\kappa}\ln\frac{\beta_\kappa + \frac12}{\beta_\kappa - \frac12}
\end{equation}
is always positive. Since $\alpha_\kappa$ is a minimizer, it follows that for almost every $p \in \T^3$ we have
$$\alpha_\kappa(p) \le 0,$$
because otherwise, by direct computation, one can check that replacing $\alpha_\kappa$ with $-\alpha_\kappa$ would decrease the energy. Now consider the set $\{p \colon \alpha_\kappa(p) = 0\}$. If this set has non-zero measure, then from \eqref{ELalpha} we get equality
\begin{equation}
\label{alpha-rho}
- U\int_{\T^3} \alpha_\kappa(q)dq - U\rho_{0,\kappa} = 0 \Longleftrightarrow \int_{\T^3} \alpha_\kappa(q)dq = -\rho_{0,\kappa}
\end{equation}
and therefore for almost every $p \in \T^3$ we get
$$T\frac{\alpha_\kappa}{\beta_\kappa}\ln\frac{\beta_\kappa + \frac12}{\beta_\kappa - \frac12} = 0.$$
By positivity of the expression \eqref{expression} this implies that $\alpha_\kappa \equiv 0$. We proved that $\alpha_\kappa$ is either zero or non-zero almost everywhere. From this it also follows that $\alpha_\kappa \equiv 0$ is equivalent to $\rho_{0,\kappa} = 0$.

Next suppose $\rho_{0,\kappa} > 0$. If equality \eqref{alpha-rho} did hold, then, by previous considerations, this would imply $\alpha_\kappa \equiv 0$ and led to contradiction. This shows that in this case $\alpha_\kappa$ is non-zero (hence negative) almost everywhere and
$$- U\int_{\T^3} \alpha_\kappa(q)dq - U\rho_{0,\kappa} < 0,$$
so
$$\|\alpha_\kappa\|_1 < \rho_{0,\kappa}.$$
\end{proof}
\begin{remark}
\label{fixed_0}
The same argument as in the above proof is also valid for the minimization problem with fixed $\rho_0$.
\end{remark}

In the following steps we will further separate the problem into two sub-cases depending on the value of the chosen minimizing $\rho_0^{\min}$ of the unrestricted functional (cf. Proposition \ref{rho_0_conv}): the first case is  $\rho_0^{\min} >0$, while the second one is $\rho_0^{\min} = 0$.

\subsubsection{Non-zero condensate density case}
Let us introduce the notation for the non-entropy terms appearing in the variational derivatives \eqref{variational_derivatives}:
\begin{equation}
\label{AB_intro}
\begin{split}
   A_\kappa(p) &:= \frac1{T}\left(\eps(p) - \mu + 2U\int_{\T^3} \gamma_\kappa(q)dq + 2U\rho_{0,\kappa}\right)\\
    B_\kappa &:= \frac1{T}\left(U\int_{\T^3} \alpha_\kappa(q)dq + U\rho_{0,\kappa}\right). 
\end{split}
\end{equation}
For $p \in \{p \colon \gamma_\kappa(p) < \kappa\},$ using the Euler-Lagrange equations (cf. Corollary \ref{EL+}) we get 
$$TA_\kappa(p) - T \frac{\gamma_\kappa + \frac12}{\beta_\kappa}\ln\frac{\beta_\kappa + \frac12}{\beta_\kappa - \frac12} = 0$$
and
$$TB_\kappa + T \frac{\alpha_\kappa}{\beta_\kappa}\ln\frac{\beta_\kappa + \frac12}{\beta_\kappa - \frac12} = 0,$$
which implies an equality
$$\sqrt{A^2_\kappa(p) - B^2_\kappa} = \ln\frac{\beta_\kappa(p) + \frac12}{\beta_\kappa(p) - \frac12},$$
which is equivalent to
$$\beta_\kappa(p) = \frac12 \frac{e^{\sqrt{A^2_\kappa(p) - B^2_\kappa}} + 1}{e^{\sqrt{A^2_\kappa(p) - B^2_\kappa}} - 1}.$$
It follows that
\begin{equation}
\label{gamma_eq}
\gamma_\kappa(p) = \frac12 \left[\frac{e^{\sqrt{A^2_\kappa(p) - B^2_\kappa}}+1}{\sqrt{A^2_\kappa(p) - B^2_\kappa}(e^{\sqrt{A^2_\kappa(p) - B^2_\kappa}}-1)}A_\kappa(p) - 1\right].  
\end{equation}
Notice that since $A_\kappa$ is a continuous function of $p$, this means $\gamma_\kappa$ is continuous on this set as well. We will estimate $\gamma_\kappa$ from above by estimating $A_\kappa(p)$ and $\sqrt{A^2_\kappa(p) - B^2_\kappa}$.

\begin{lemma}
\label{Cc}
For sufficiently large $\kappa$, there exist constants $C$ and $c$ independent of $\kappa$ such that for all $ p \in \T^3$ we have
$$A_\kappa(p) \le C  $$
and
$$\sqrt{A^2_\kappa(p) - B^2_\kappa} \ge c.$$
\end{lemma}

\begin{proof}
Using the $\rho_0$ derivative condition (cf. third relation in Corollary \ref{EL+} which now, as we assume non-zero condensate density, holds with equality) we get the relation
\begin{equation}
\label{A_rewrite}
\mu = U\int_{\T^3} \alpha_\kappa(q)dq + 2U\int_{\T^3} \gamma_\kappa(q)dq + U\rho_{0,\kappa}.  
\end{equation}
With this we can rewrite expression for $A_\kappa$ as
$$A_\kappa(p) = \frac1{T}\left(\eps(p) - U\int_{\T^3} \alpha_\kappa(q)dq + U\rho_{0,\kappa}\right),$$
in particular we got rid of the $\mu$ dependence in the entire expression. Now, using Proposition \ref{rest_stucture} we can estimate $A_\kappa$ from above by
$$\left|\int_{\T^3} \alpha_\kappa(q)dq\right| = \|\alpha_\kappa\|_1 \le \rho_{0,\kappa}.$$
As the minimizing sequence  $\rho_{0,\kappa}$ is bounded (cf. Corollary \ref{bounded_seq}) and $\varepsilon(p)\leq 12$,  we have
$$A_\kappa(p) \le \frac1{T}\left(12+ UC_1\right) =: C = C(U,T).$$
To estimate $\sqrt{A^2_\kappa(p) - B^2_\kappa}$ we can write
$$\sqrt{A^2_\kappa(p) - B^2_\kappa} = \sqrt{(A_\kappa(p) - B_\kappa)(A_\kappa(p) + B_\kappa)}$$
and estimate each factor. We have
$$A_\kappa(p) - B_\kappa = \frac1{T}\left(\eps(p) -2U\int_{\T^3} \alpha_\kappa(q) dq\right) \ge -\frac{2U}{T}\int_{\T^3} \alpha_\kappa(q) dq = \frac{2U}{T}\|\alpha_\kappa\|_1$$
and
$$A_\kappa(p) + B_\kappa = \frac1{T}\left(\eps(p) + 2U\rho_{0,\kappa}\right)  \ge \frac{2U}{T}\rho_{0,\kappa}.$$
Our goal is to estimate the right hand sides of those inequalities by a constant independent of $\kappa$.

Since we assumed that $\rho_0^{\min}$ minimizing the unrestricted problem is positive, by Proposition \ref{rho_0_conv} for sufficiently large $\kappa$ we have
$$\rho_{0,\kappa} \ge \frac12\rho_0^{\min} > 0.$$
For $p \in \{p \colon \gamma_\kappa(p) < \kappa\}$ we can combine first two equalities from Corollary \ref{EL+} to obtain
\begin{align*}
    A_\kappa(p) + \frac{\gamma_\kappa(p) + \frac12}{\alpha_\kappa(p)} B_\kappa = 0.
\end{align*}
We have used the fact that $\alpha_\kappa(p) \ne 0$ almost everywhere (cf. Proposition \ref{rest_stucture}). It follows that
\begin{align*}
    \gamma_\kappa(p) &= -\frac{A_\kappa(p)}{B_\kappa} \alpha_\kappa(p) - \frac12,
\end{align*}
where, once again, by Proposition \ref{rest_stucture} we know that $B_\kappa \ne 0$. Writing the expressions for $A_\kappa(p)$ and $B_\kappa$ explicitly and using the fact $\alpha_\kappa(p) < 0$ we have
$$\gamma_\kappa(p) = \frac{\eps(p) + U\int_{\T^3} |\alpha_\kappa(q)|dq + U\rho_{0,\kappa}}{-U\int_{\T^3} |\alpha_\kappa(q)|dq + U\rho_{0,\kappa}} |\alpha_\kappa(p)| - \frac12.$$
Now we can use the fact that for sufficiently large $\kappa$ we have $\gamma_\kappa(p) \ge c$ almost everywhere (cf. Lemma \ref{bound_below}). We get
$$ c \le \frac{\eps(p) + U\int_{\T^3} |\alpha_\kappa(q)|dq + U\rho_{0,\kappa}}{-U\int_{\T^3} |\alpha_\kappa(q)|dq + U\rho_{0,\kappa}} |\alpha_\kappa(p)| - \frac12$$
and after rewriting
$$\frac{(c+\frac12)(U\rho_{0,\kappa} -U\int_{\T^3} |\alpha_\kappa(q)|dq)}{\eps(p) + U\int_{\T^3} |\alpha_\kappa(q)|dq + U\rho_{0,\kappa}} \le |\alpha_\kappa(p)|.$$
Denoting $c' = c+\frac12$ and estimating $\eps(p) \le 12$ and $\|\alpha_\kappa\|_1 \le \rho_{0,\kappa}$ (once again we are using Proposition \ref{rest_stucture}) we get
$$\frac{c'(U\rho_{0,\kappa} -U\int_{\T^3} |\alpha_\kappa(q)|dq)}{12 + 2U\rho_{0,\kappa}} \le |\alpha_\kappa(p)|.$$
Taking integral of both sides over the set $\{p \colon \gamma_\kappa(p) < \kappa\}$ we obtain
\begin{equation}
\label{alpha_point_est}
|\{p \colon \gamma_\kappa(p) < \kappa\}|\cdot \frac{c'(U\rho_{0,\kappa} -U\|\alpha_\kappa\|_1)}{12 + 2U\rho_{0,\kappa}} \le \int_{\{p \colon \gamma_\kappa(p) < \kappa\} }|\alpha_\kappa(p)|dp.    
\end{equation}
We can trivially estimate the right hand side
$$\int_{\{p \colon \gamma_\kappa(p) < \kappa\} }|\alpha_\kappa(p)|dp \le \|\alpha_\kappa\|_1.$$
As for the left hand side we will estimate the measure of the set $\{p \colon \gamma_\kappa(p) < \kappa\}$. Since $\gamma_\kappa$ are uniformly bounded in $L^1$, i.e. $\|\gamma_\kappa\|_1 \le C$ for some constant $C$ independent of $\kappa$, using that $\gamma_\kappa \ge 0$ we have
$$C \ge \|\gamma_\kappa\|_1 = |\{p \in \T^3 \colon \gamma_{\kappa}(p) = \kappa\}|\cdot \kappa + \int_{\{\gamma_\kappa < \kappa\}} \gamma_\kappa(p) dp$$
and therefore
$$|\{p \in \T^3 \colon \gamma_{\kappa}(p) = \kappa\}| \le \frac{C}{\kappa}.$$
It follows that for sufficiently large $\kappa$ we have
$$|\{p \in \T^3 \colon \gamma_{\kappa}(p) = \kappa\}| \le \frac12$$
and hence (recall that the measure of the whole space $\T^3$ is equal to one)
$$|\{p \in \T^3 \colon \gamma_{\kappa}(p) < \kappa\}|\ge \frac12.$$
Combining those facts in \eqref{alpha_point_est} gives us
$$\frac{1}{2} \cdot \frac{c'(U\rho_{0,\kappa} -U\|\alpha_\kappa\|_1)}{12 + 2U\rho_{0,\kappa}}\le \|\alpha_\kappa\|_1 .$$
Solving this inequality for $\|\alpha_\kappa\|_1$ yields
$$\frac{c'U\rho_{0,\kappa}}{24 + 4U\rho_{0,\kappa} + c'U} \le \|\alpha_\kappa\|_1.$$
By convergence of $\rho_{0,\kappa}$ (cf. Proposition \ref{rho_0_conv}), for sufficiently large $\kappa$, we get the inequalities
$$\frac12 \rho_{0,\kappa} \le \rho_0^{\min} \le 2\rho_{0,\kappa},$$
so we can estimate $\|\alpha_\kappa\|_1$ by a constant $\tilde c$ independent of $\kappa$, i.e., 
$$\tilde c := \frac{\frac12c'U\rho_{0}^{\min}}{24 + 8U\rho_{0}^{\min} + c'U} \le \|\alpha_\kappa\|_1.$$
Therefore we get the desired estimates
$$A_\kappa(p) - B_\kappa \ge \frac{2U\tilde c}{T}$$
and
$$A_\kappa(p) + B_\kappa \ge \frac{U\rho_0^{\min}}{T}.$$
Now we can estimate
$$\sqrt{A^2_\kappa(p) - B^2_\kappa} \ge \sqrt{\frac{2U\tilde c}{T} \cdot \frac{U\rho_0^{\min}}{T}} =: c > 0.$$
\end{proof}
We can use the above result to obtain an important corollary.
\begin{cor}
\label{K_bound}
For sufficiently large $\kappa$, for $p \in \{p \colon \gamma_\kappa(p) < \kappa\}$ we have
$$\gamma_\kappa(p) \le K$$
for some constant $K$ independent of $\kappa$.
\end{cor}
\begin{proof}
Using equation \eqref{gamma_eq} and estimates obtained in the previous proposition we get
$$\gamma_\kappa(p) = \frac12 \left[\frac{e^{\sqrt{A^2_\kappa(p) - B^2_\kappa}}+1}{\sqrt{A^2_\kappa(p) - B^2_\kappa}(e^{\sqrt{A^2_\kappa(p) - B^2_\kappa}}-1)}A_\kappa(p) - 1\right] \le \frac12 \left[\frac{e^{C}+1}{c(e^{c}-1)}C - 1\right] =: K.$$
\end{proof}
Now we will show that $\gamma_\kappa$ is actually bounded by $K$ on the whole domain.

\begin{lemma}
\label{direct_def}
For sufficiently large $\kappa$ (precisely: for $\kappa > K$) the set $\{p \colon \gamma_\kappa(p) = \kappa\}$ has measure zero (in fact, it is empty). In particular $\gamma_\kappa$ is determined by \eqref{gamma_eq} for all $p \in \T^3$ and satisfies inequality
$$\gamma_\kappa(p) \le K.$$
\end{lemma}

\begin{proof}
We can decompose the domain $\T^3$ into two disjoint sets and use the previous corollary to obtain
$$\T^3 = \{p \colon \gamma_\kappa(p) = \kappa\} \cup \{p \colon \gamma_\kappa(p) < \kappa\} =\{p \colon \gamma_\kappa(p) = \kappa\} \cup \{p \colon \gamma_\kappa(p) \le K\}.$$
First recall that we obtained equation \eqref{gamma_eq} for $\gamma_\kappa$ by using Euler-Lagrange equations
\begin{equation}
\label{ABeq}
\begin{split}
&A_\kappa(p) -  \frac{\gamma_\kappa + \frac12}{\beta_\kappa}\ln\frac{\beta_\kappa + \frac12}{\beta_\kappa - \frac12} = 0,
\\
&B_\kappa + \frac{\alpha_\kappa}{\beta_\kappa}\ln\frac{\beta_\kappa + \frac12}{\beta_\kappa - \frac12} = 0,\\
\end{split}  
\end{equation}
where $A_\kappa(p)$ and $B_\kappa$ are given in \eqref{AB_intro}.
%$$A_\kappa(p) = \frac1{T}\left(\eps(p) - U\int_{\T^3} \alpha_\kappa(q)dq + U\rho_{0}\right)$$
%and
%$$B_\kappa = \frac1{T}\left(U\int_{\T^3} \alpha_\kappa(q)dq + %U\rho_0\right).$$
In fact, equations \eqref{ABeq} are equivalent to \eqref{gamma_eq} and an analogous equation for $\alpha_\kappa$, that is if we define
\begin{equation}
\label{ABeq2}
\begin{split}
\tilde\gamma_\kappa(p) &= \frac12 \left[\frac{e^{\sqrt{A^2_\kappa(p) - B^2_\kappa}}+1}{\sqrt{A^2_\kappa(p) - B^2_\kappa}(e^{\sqrt{A^2_\kappa(p) - B^2_\kappa}}-1)}A_\kappa(p) - 1\right]\\
\tilde \alpha_\kappa(p) &= -\frac12 \frac{e^{\sqrt{A^2_\kappa(p) - B^2_\kappa}}+1}{\sqrt{A^2_\kappa(p) - B^2_\kappa}(e^{\sqrt{A^2_\kappa(p) - B^2_\kappa}}-1)}B_\kappa
\end{split}  
\end{equation}
then equations \eqref{ABeq} are satisfied with $\gamma_\kappa$ replaced by $\tilde \gamma_\kappa$ and $\alpha_\kappa$ replaced by $\tilde \alpha_\kappa$ (note that $A_\kappa$ and $B_\kappa$ still depend on the integrals of the restricted minimizers $\gamma_\kappa$ and $\alpha_\kappa$). 

The functions $\tilde \gamma_\kappa$ and $\tilde \alpha_\kappa$ are well defined on the whole domain $\T^3$ as the bounds derived in Lemma \ref{Cc} hold for every $p$. Furthermore, for every $p \in \T^3$ function $\tilde \gamma_\kappa$ satisfies the inequality
$$\tilde \gamma_\kappa(p) \le K$$
and, as $\tilde \gamma_\kappa$ and $\tilde \alpha_\kappa$ satisfy the Euler-Lagrange equations, we also have
%Now notice that bounds on $A_\kappa(p)$ and $\sqrt{A_\kappa^2(p) - B_\kappa^2}$ from Proposition \ref{Cc} are well defined for every $p \in \T^3$ (as it is a corollary from Euler-Lagrange equations), so we can define functions $\tilde \gamma$ and $\tilde \alpha$ on the whole domain $\T^3$. In particular this means that for every $p \in \T^3$ we have
\begin{align*}
 A_\kappa(p) -  \frac{\tilde\gamma_\kappa + \frac12}{\tilde\beta_\kappa}\ln\frac{\tilde\beta_\kappa + \frac12}{\tilde\beta_\kappa - \frac12} = 0.
\end{align*}
The absolute value of entropy derivative with respect to $\gamma$, treated as the function of variables $(\gamma, \alpha)$, is strictly decreasing in $\gamma$. It follows that for $p \in \{p \colon \gamma_\kappa(p) = \kappa\}$, for $\kappa > K$ we have strict inequality
$$\tilde \gamma(p) < \gamma_\kappa(p),$$
so on this set we have
$$\frac{\partial \F}{\partial \gamma}(p) = TA_\kappa(p) -  T\frac{\gamma_\kappa + \frac12}{\beta_\kappa}\ln\frac{\beta_\kappa + \frac12}{\beta_\kappa - \frac12} > 0.$$
Thus the set $\{p \colon \gamma_\kappa(p) = \kappa\}$ has to be of zero measure, because as $\gamma_\kappa$ is a minimizer, it must satisfy the reverse inequality for almost every $p \in \T^3$. 
\end{proof}

We are ready to prove existence of minimizers of the unrestricted problem in the case with $\rho_0^{\min} > 0$. This is partially proving Theorem \ref{main_existence}.
\begin{prop}
\label{T>0existence1}
For $T>0$ and minimizing $\rho_0^{\min}>0$ there exists a minimizer of the unrestricted problem. 
\end{prop}
\begin{proof}
We will use the fact that $(\gamma_\kappa, \alpha_\kappa, \rho_{0,\kappa})$ forms a minimizing sequence and, for $\kappa > K$, we have pointwise inequality
$$\gamma_\kappa(p) \le K.$$
Let $\kappa_0$ be the smallest $\kappa$ such that $\kappa_0 > K$. Observe that for $\kappa \ge \kappa_0$, by the above inequality, we have $(\gamma_\kappa, \alpha_\kappa, \rho_{0,\kappa}) \in \D_{\kappa_0}$. Furthermore, we can assume
$\F(\gamma_\kappa, \alpha_\kappa, \rho_{0,\kappa})$ is decreasing as otherwise we can pass to a suitable subsequnce. It follows that
\begin{align*}
    \inf_\D \F(\gamma, \alpha, \rho_0) &= \lim_{\kappa \to \infty}\F(\gamma_\kappa, \alpha_\kappa, \rho_{0,\kappa})
    = \inf_\kappa\F(\gamma_\kappa, \alpha_\kappa, \rho_{0,\kappa})
    \\&\ge  \inf_{\D_{\kappa_0}}\F(\gamma, \alpha, \rho_0)
    = \F(\gamma_{\kappa_0}, \alpha_{\kappa_0}, \rho_{0,\kappa_0})
    \ge \inf_\D \F(\gamma, \alpha, \rho_0),
\end{align*}
which proves that $(\gamma_{\kappa_0}, \alpha_{\kappa_0}, \rho_{0,\kappa_0})$ is the minimizer of the unrestricted problem. 
\end{proof}

\subsubsection{Zero condensate density case}

In this section we assume that $\rho_0^{\min}$ minimizing the function $f$ (recall Proposition \ref{rho_0_conv}) is equal to zero. Recall also Proposition \ref{restriced_existence} from which it follows that there exist minimizers of the restricted problem with fixed $\rho_0 = 0$. By Remark \ref{fixed_0} we already know $\alpha_\kappa \equiv 0$. In the next lemmas, using the Euler-Lagrange inequalities from Corollary \ref{EL+} for $\gamma_\kappa$, we will directly show that the sequence $\gamma_\kappa$ of the restricted minimizers up to a subsequence converges pointwise almost everywhere and in $L^1$ to some function $\tilde \gamma$.

\begin{lemma}
\label{pointwise}
There exists a subsequence of $\gamma_\kappa$ convergent pointwise almost everywhere on $\T^3$ such that there exist limits
$$\lim_{\kappa \to \infty} \int_{\T^3} \gamma_\kappa(p)dp =: G $$
and
$$\lim_{\kappa \to \infty} \gamma_\kappa(p) = \tilde \gamma (p) = \left(e^{\frac{\eps(p) - \mu + 2UG}{T}} - 1\right)^{-1}.$$
\end{lemma}

\begin{proof} \textit{Step 1.} Because $\gamma_\kappa$ is a minimizing sequence, by Corollary \ref{bounded_seq} this sequence is bounded in $L^1(\T^3)$, so we can pass to a subsequence such that sequence of integrals of $\gamma_\kappa$ is convergent. We will assume this in further steps.

\noindent \textit{Step 2.} Using the notation from the previous section and repeating the computations we get the same result as in \eqref{gamma_eq}, that is on the set $\{p \colon \gamma_\kappa(p) < \kappa\}$ we have
$$\gamma_\kappa(p) = \frac12 \left[\frac{e^{\sqrt{A^2_\kappa(p) - B^2_\kappa}}+1}{\sqrt{A^2_\kappa(p) - B^2_\kappa}(e^{\sqrt{A^2_\kappa(p) - B^2_\kappa}}-1)}A_\kappa(p) - 1\right].$$
In this case we have $\rho_0^{\min} = 0$ and $\alpha_\kappa \equiv 0$, so
$$B_\kappa = 0$$
and
$$A_\kappa(p) = \frac1{T}\left(\eps(p) - \mu + 2U\int_{\T^3} \gamma_\kappa(q)dq\right),$$
therefore on the aforementioned set we have
\begin{align*}
\gamma_\kappa(p) &= \frac12 \left[\frac{e^{A_\kappa(p)}+1}{A_\kappa(p)(e^{A_\kappa(p)}-1)}A_\kappa(p) - 1\right] = \frac{1}{e^{A_\kappa(p)} - 1} = \left(e^{\frac1{T}\left(\eps(p) - \mu + 2U\int_{\T^3}\gamma_\kappa(q)dq\right)} - 1\right)^{-1}.
\end{align*}

\noindent \textit{Step 3.} We will give an explicit description of the set $\{p \colon \gamma_\kappa(p) < \kappa\}$. We define
$$E_\kappa := \left\{p \colon \eps(p) > T\ln(1 + 1/\kappa) + \mu - 2U\int_{\T^3} \gamma_\kappa(q) dq\right\}$$
and claim that
\begin{equation}
\label{explicit_set}
\{p \colon \gamma_\kappa(p) < \kappa\} =  E_\kappa.   
\end{equation}
In order to prove this, first note that the condition in the definition of set $E_k$ is equivalent to the cutoff:
$$\eps(p) > T\ln(1 + 1/\kappa) + \mu - 2U\int_{\T^3} \gamma_\kappa(q) dq \Longleftrightarrow \left(e^{\frac{1}{T}\left(\eps(p) - \mu + 2U\int_{\T^3}\gamma_\kappa(q)dq\right)} - 1\right)^{-1} < \kappa,$$
hence we certainly have an inclusion
$$E_\kappa \subset \{p \colon \gamma_\kappa(p) < \kappa\}.$$
To prove that this is in fact an equality of sets, we will use similar argument as in the proof of Proposition \ref{direct_def} -- define $\tilde \gamma_\kappa(p)$ as
$$\tilde \gamma_\kappa(p) = \left\{ 
\begin{array}{ll}
     \left(e^{\frac1{T}\left(\eps(p) - \mu + 2U\int_{\T^3}\gamma_\kappa(q)dq\right)} - 1\right)^{-1} &\text{ if } p \in E_\kappa, \\
     \kappa &\text{ if } p \in \T^3 \setminus E_\kappa.
\end{array}\right.$$
Note that this function belongs to the restricted domain (i.e. $\tilde \gamma_\kappa \le \kappa$) and satisfies the Euler-Lagrange inequalities as in Corollary \ref{EL+}. Furthermore, for almost every $p \in \T^3$ we have
$$\gamma_\kappa(p) \ge \tilde \gamma_\kappa(p),$$
as $\tilde \gamma_\kappa$ is a modification of $\gamma_\kappa$ only for the set of $p$'s in $E_\kappa$ with $\gamma_\kappa(p) = \kappa$. In particular, for such $p$'s we have the strict inequality $\tilde \gamma_\kappa(p) < \gamma_\kappa(p) = \kappa$. If the aforementioned set had positive measure then, using the monotonicity of the entropy in $\gamma$, we would get
$$\frac{\partial \F}{\partial \gamma} = \eps(p) - \mu + 2U\int_{\T^3} \gamma_\kappa(q)dq - T \ln\left(1 + \frac{1}{\gamma_\kappa}\right) > 0,$$
which contradicts the fact $\gamma_\kappa$, as a minimizer, has to satisfy reverse inequality almost everywhere. This finishes the proof of the claim \eqref{explicit_set}. 

\noindent \textit{Step 4.} From the previous step we also got the relation
\begin{equation}\label{eq:TsetminusEk}
\{p \colon \gamma_\kappa(p) = \kappa\} = \T^3 \setminus E_\kappa = \left\{p \colon \eps(p) \le T\ln(1 + 1/\kappa) + \mu - 2U\int_{\T^3} \gamma_\kappa(q) dq\right\}.
\end{equation}
Recall the estimate on the measure of the set $\{p \in \T^3 \colon \gamma_{\kappa}(p) = \kappa\}$ that we have already used in the proof of Lemma \ref{Cc}:
$$|\{p \in \T^3 \colon \gamma_{\kappa}(p) = \kappa\}| \le \frac{C}{\kappa}.$$
In particular, the measure of this set goes to zero as $\kappa \to \infty$. We can conclude that
$$G = \lim_{\kappa \to \infty} \int_{\T^3} \gamma_\kappa(q)dq \ge \frac{\mu}{2U}.$$
Indeed,  if $G < \frac{\mu}{2U}$, then there exists $\eta > 0$ such that for sufficiently large $\kappa$ we have
$$T\ln(1 + 1/\kappa) + \mu - 2U\int_{\T^3} \gamma_\kappa(q) dq > \eta,$$
so each set $\T^3 \setminus E_\kappa = \left\{p \colon \eps(p) \le T\ln(1 + 1/\kappa) + \mu - 2U\int_{\T^3} \gamma_\kappa(q) dq\right\}$ is an open neighborhood of $p=0$ (by continuity of $\eps(p)$) whose measure is positive and separated from zero, which contradicts convergence of this measure to zero.

\noindent \textit{Step 5.} The conclusion that the measure of $\T^3 \setminus E_\kappa$ goes to zero allows us to pass to a subsequence such that the measures of the sets $\T^3 \setminus E_\kappa$ are non-increasing in $\kappa$, which in particular implies that the sequence $T\ln(1 + 1/\kappa) + \mu - 2U\int_{\T^3} \gamma_\kappa(q) dq$ is non-increasing and thus, by the definition of the set in \eqref{eq:TsetminusEk}, the sequence of sets $\T^3 \setminus E_\kappa$ is descending (and so the sequence $E_\kappa$ is ascending). It also follows that the intersection $\bigcap_{\kappa} \T^3 \setminus E_\kappa$ has measure zero (this intersection could possibly be empty) and therefore
$$\bigcap_{\kappa} \T^3 \setminus E_\kappa \subset \{0\}.$$ 

\noindent \textit{Step 6.} Now we can prove $\gamma_\kappa$ converges pointwise almost everywhere. Fix $p \in \T^3 \setminus \{0\}$. From previous step we deduce that there exist $\kappa(p)$ such that for all $\kappa > \kappa(p)$ we have $p \in E_\kappa$, so 
$$\gamma_\kappa(p) = \left(e^{\frac1{T}\left(\eps(p) - \mu + 2U\int_{\T^3}\gamma_\kappa(q)dq\right)} - 1\right)^{-1}.$$
It follows that passing to the limit is possible and yields
$$\lim_{\kappa \to \infty} \gamma_\kappa(p) = \left(e^{\frac1{T}\left(\eps(p) - \mu + 2UG\right)} - 1\right)^{-1} = \tilde \gamma(p).$$
This result holds for any $p \ne 0$, so indeed $\gamma_\kappa$ converges pointwise almost everywhere.
\end{proof}

\begin{lemma}
\label{zero_difference}
The sequence $\gamma_\kappa$ converges to $\tilde \gamma$ in $L^1$.
\end{lemma}
\begin{proof}
Since $\gamma_\kappa$ are positive functions, it suffices to show that
$$\int_{\T^3} \tilde \gamma(p) dp = \lim_{\kappa \to \infty} \int_{\T^3} \gamma_\kappa(p) dp$$
that is
$$\int_{\T^3} \left(e^{\frac1{T}\left(\eps(p) - \mu + 2UG\right)} - 1\right)^{-1} dp = G.$$
By the Fatou lemma we already have the inequality
$$\int_{\T^3} \left(e^{\frac1{T}\left(\eps(p) - \mu + 2UG\right)} - 1\right)^{-1} dp \le G.$$
Suppose the equality does not hold, that is we have strict inequality
$$\int_{\T^3} \left(e^{\frac1{T}\left(\eps(p) - \mu + 2UG\right)} - 1\right)^{-1} dp < G.$$
As in \cite{NapReuSol1}, we will make an argument based on the idea that in this case we could strictly decrease the energy by adding mass to the condensate, which contradicts the fact $\rho_0 = 0$ is a minimizer. Denote
$$D := G - \int_{\T^3} \left(e^{\frac1{T}\left(\eps(p) - \mu + 2UG\right)} - 1\right)^{-1} dp  > 0.$$
We will be working with the same sequence of $\kappa$'s as in the proof of the previous Lemma \ref{pointwise}.

\noindent \textit{Step 1.} For fixed set $E_{K}$ and $p \in E_{K}$ for $\kappa$ sufficiently large (depending on $K$), we have
$$\gamma_\kappa(p) = \left(e^{\frac1{T}\left(\eps(p) - \mu + 2U\int_{\T^3}\gamma_\kappa(q)dq\right)} - 1\right)^{-1} \le \left(e^{\frac1{T}\left(\eps(p) - \mu + 2UM\right)} - 1\right)^{-1},$$
where, using the conclusions of Step 5. from the proof of previous proposition
$$M = \max \left\{\int_{E_{K}} \gamma_{K}(q) dq, \frac{\mu}{2U} \right\}.$$
The above function estimating $\gamma_\kappa$ is integrable on the set $E_{K}$, so using Dominated Convergence Theorem we get
$$\lim_{\kappa \to \infty}\int_{E_{K}} \gamma_\kappa(p) dp = \int_{E_{K}} \tilde \gamma(p) dp.$$
Now, using the fact $E_\kappa$ is an ascending family of sets and the measure of the set $\T^3 \setminus \bigcup_{\kappa}E_\kappa$ is zero, we get
$$\lim_{K \to \infty}\int_{E_K} \tilde \gamma(p) dp = \int_{\T^3} \tilde \gamma(p) dp = G - D.$$
We can use diagonal method to extract a sequence $\kappa(K)$ such that
$$\lim_{K \to \infty} \int_{E_K} \gamma_{{\kappa(K)}}(p) dp = G - D.$$
We also have
$$\lim_{K \to \infty} \int_{\T^3 \setminus E_K} \gamma_{{\kappa(K)}}(p) dp = D.$$

\noindent \textit{Step 2.} We will now evaluate the BBH functional $\F$ on the triple $(\gamma_{\kappa(K)},0,0)$. Separating the integration into sets $E_\kappa$ and its complement and estimating some positive terms from below by zero we obtain
\begin{equation}
\label{big_est}
\begin{split}
\F(\gamma_{\kappa(K)}, \alpha \equiv 0, \rho_0 = 0) &= \int_{\T^3}\eps(p)\gamma_{\kappa(K)}(p) dp - \mu\int_{\T^3} \gamma_{\kappa(K)} (p) dp \\&- T\int_{\T^3} \big[(\gamma_{\kappa(K)}(p) + 1)\ln(\gamma_{\kappa(K)}(p) +1) - \gamma_{\kappa(K)}(p)\ln \gamma_\kappa(p)\big] dp 
\\& + U\left(\int_{\T^3} \gamma_{\kappa(K)} (p) dp\right)^2
\\& \ge \int_{E_K}\eps(p)\gamma_{\kappa(K)}(p) dp - \mu\left(\int_{E_K} \gamma_{\kappa(K)} (p) dp  + \int_{\T^3 \setminus E_K} \gamma_{\kappa(K)} (p) dp\right)
\\&- T\int_{E_K} \big[(\gamma_{\kappa(K)}(p) + 1)\ln(\gamma_{\kappa(K)}(p) +1) - \gamma_{\kappa(K)}(p)\ln \gamma_{\kappa(K)}(p)\big] dp
\\&- T\int_{\T^3 \setminus E_K} \big[(\gamma_{\kappa(K)}(p) + 1)\ln(\gamma_{\kappa(K)}(p) +1) - \gamma_{\kappa(K)}(p)\ln \gamma_{\kappa(K)}(p)\big] dp
\\&+ U\left(\int_{E_K} \gamma_{\kappa(K)} (p) dp\right)^2 + 2U\left(\int_{E_K} \gamma_{\kappa(K)} (p) dp\right) \left(\int_{\T^3 \setminus E_K} \gamma_{\kappa(K)} (p) dp\right)
\\&+ U\left(\int_{\T^3 \setminus E_K} \gamma_{\kappa(K)} (p) dp\right)^2.
\end{split}
\end{equation}
Let us now estimate the entropy term on the set $\T^3 \setminus E_K$. Estimating the integrand on the set where $\gamma_{\kappa(K)}(p) \le 1$ we get
$$(\gamma_{\kappa(K)}(p) + 1)\ln(\gamma_{\kappa(K)}(p) +1) - \gamma_{\kappa(K)}(p)\ln \gamma_{\kappa(K)}(p) \le 2\ln2 + 1 = c $$
and on the set where $\gamma_{\kappa(K)}(p) > 1$
\begin{align*}
&(\gamma_{\kappa(K)}(p) + 1)\ln(\gamma_{\kappa(K)}(p) +1) - \gamma_{\kappa(K)}(p)\ln \gamma_{\kappa(K)}(p) 
\\&= (\gamma_{\kappa(K)}(p) + 1) \ln \left(1 + \frac1{\gamma_{\kappa(K)}(p)}\right) + \ln \gamma_{\kappa(K)}(p) \\& \le 1 + \frac1{\gamma_{\kappa(K)}(p)} + \ln \gamma_{\kappa(K)}(p) 
 \\& \le 2 + \gamma^{1/2}_{\kappa(K)} (p).
\end{align*}
Using this, Cauchy-Schwarz inequality and uniform bound $\|\gamma_\kappa\|_1 \le C$ we get
\begin{align*}
&T\int_{\T^3 \setminus E_K} \big[(\gamma_{\kappa(K)}(p) + 1)\ln(\gamma_{\kappa(K)}(p) +1) - \gamma_{\kappa(K)}(p)\ln \gamma_{\kappa(K)}(p)\big] dp 
\\&\le T\int_{\T^3 \setminus E_K} c + \gamma_{\kappa(K)}^{1/2}(p) dp \\& \le C|\T^3 \setminus E_K| + |\T^3 \setminus E_K|^{1/2}\|\gamma_{\kappa(K)}\|_1^{1/2}
\\& \le C|\T^3 \setminus E_K|^{1/2} \to 0 \text{ as } K \to \infty.
\end{align*}

\noindent \textit{Step 3.} Define a function
$$g_K(p) = \gamma_{\kappa(K)}(p) \mathbbm{1}_{E_K}(p).$$
Using \eqref{big_est} we can rewrite it as
\begin{align*}
\F(\gamma_{\kappa(K)}, 0, 0) &\ge \F(g_K, 0, D) + \mu D - \mu \int_{\T^3 \setminus E_K} \gamma_{\kappa(K)}(p) 
\\& - T\int_{\T^3 \setminus E_K} \big[(\gamma_{\kappa(K)}(p) + 1)\ln(\gamma_{\kappa(K)}(p) +1) - \gamma_{\kappa(K)}(p)\ln \gamma_{\kappa(K)}(p)\big] dp
\\& + U\left(\int_{\T^3 \setminus E_K} \gamma_{\kappa(K)} (p) dp\right)^2 + 2U\left(\int_{E_K} \gamma_{\kappa(K)} (p) dp\right) \left(\int_{\T^3 \setminus E_K} \gamma_{\kappa(K)} (p) dp\right)
\\&- \frac12UD^2 - 2UD\int_{E_K} \gamma_{\kappa(K)}(p)dp
\\&= \F(g_K, 0, D) + \mu \left[D - \int_{\T^3 \setminus E_K} \gamma_{\kappa(K)}(p) \right] 
\\& - T\int_{\T^3 \setminus E_K} \big[(\gamma_{\kappa(K)}(p) + 1)\ln(\gamma_{\kappa(K)}(p) +1) - \gamma_{\kappa(K)}(p)\ln \gamma_{\kappa(K)}(p)\big] dp
\\& + U\bigg[\left(\int_{\T^3 \setminus E_K} \gamma_{\kappa(K)} (p) dp\right)^2 - D^2 \bigg]
\\&+ 2U\left(\int_{E_K} \gamma_{\kappa(K)} (p) dp\right)\bigg[ \left(\int_{\T^3 \setminus E_K} \gamma_{\kappa(K)} (p) dp\right) - D\bigg]
\\&+ \frac12UD^2.
\end{align*}
Due to the results of the previous steps, entropy term and expressions in the brackets converge to zero as $K$ goes to infinity. Denote all of those terms by $R(K)$ and rewrite the above inequality as
$$\F(\gamma_{\kappa(K)}, 0, 0) \ge \F(g_K, 0, D) + R(K) + \frac12 UD^2.$$
Recalling definition \eqref{double_min} of the function $f$ we can further estimate
$$\F(\gamma_{\kappa(K)}, 0, 0) \ge f(D) + R(K) + \frac12 UD^2.$$
As $\gamma_{\kappa(K)}$ is a minimizing sequence, passing to the limit $K \to \infty$ yields
$$f(\rho_0^{\min} = 0) \ge f(D) + \frac12 UD^2.$$
Therefore if $D > 0$ then we get a strict inequality
$$f(0) > f(D),$$
which contradicts the fact $\rho_0^{\min} = 0$ was a minimizer of the function $f$.
\end{proof}

Now we are ready to prove $\tilde \gamma$ is a desired minimizer of the unrestricted problem.
\begin{prop}
\label{rho_0=0_existence}
If $\rho_0^{\min} = 0$ then there exist $(\gamma^{\min}, \alpha^{\min})$ minimizing the functional $\F$. 
\end{prop}
\begin{proof}
We will show that $(\tilde \gamma, 0 , 0)$ is a desired minimizer, where $\tilde \gamma$ is a limit of $\gamma_\kappa$ from the previous lemmas. We already know that $\gamma_\kappa$ converges to $\tilde \gamma$ pointwise and in $L^1$. Just like in the proof of Proposition \ref{restriced_existence} we can estimate kinetic and entropy term using the Fatou lemma and for remaining terms use the $L^1$ convergence. This gives us
$$\F(\tilde \gamma, 0 ,0) \le \liminf_{\kappa \to \infty} \F(\gamma_\kappa, 0, 0) = \inf_{(\gamma, \alpha, \rho_0) \in\D} \F(\gamma, \alpha, \rho_0),$$
which proves $(\tilde \gamma, 0 ,0)$ is a minimizer of the unrestricted problem.
\end{proof}
\begin{remark}
\label{rem:T>0proof}
Proposition \ref{rho_0=0_existence} together with Proposition \ref{T>0existence1} prove Theorem \ref{main_existence} in the case $T>0.$ 
\end{remark}
\subsection{Zero temperature case}
Now we focus on the zero temperature situation. In the variational derivatives of the functional there is no explicit dependence on values of $\gamma$ and $\alpha$ (only on values of the integrals of those functions), so we need to use some different methods. As mentioned before, our approach to this problem will be based on \cite{Old} with appropriate modifications to the lattice setting.

First notice that if $\mu \le 0$ then functional $\F$ is non-negative and the vacuum $\gamma \equiv 0$, $\alpha \equiv 0$, $\rho_0 = 0$ is a minimizer. This corresponds to the fact in this case the functional is obtained as an expectation value of a non-negative Hamiltonian, see Appendix \ref{app:derivation}. Therefore, from now on we will assume $\mu > 0$. We will start with proving that restricted minimizers are pure states.
\begin{prop}
\label{purity}
For $T=0$ and sufficiently large $\kappa$ minimizers of the restricted problem satisfy
\begin{equation} \label{eq:pureqfT0}
\alpha_\kappa^2(p) = \gamma_\kappa(p)(\gamma_\kappa(p)+1)
\end{equation}
for almost all $p\in \T^3$. Furthermore, a minimizing $\rho_{0,\kappa}$ is positive.
\end{prop}

\begin{proof}
Denote as before $\rho_\gamma = \|\gamma\|_1$ and  let $B_\eta = \{p \colon |p| < \eta\}$. Similarly as in the proof of \cite{NapReuSol1}[Prop. 5.2.] we consider states of the form
$$\gamma(p) = \lambda \mathbbm{1}_{B_\eta}(p), \; \; \alpha(p) = -\sqrt{\gamma(p)(\gamma(p)+1)}, \; \; \rho_0 = \frac{\mu}{U} - \rho_\gamma$$
where for a given (large) $\lambda$ we choose $\eta$ small enough such that $\rho_0 > 0$. Evaluating the BBH functional $\F$ in this state (using similar form of $\F$ like in \eqref{func2}) we obtain
\begin{align*}
\F(\gamma, \alpha, \rho_0) &= \lambda \int_{B_\eta}\eps(p)dp - \frac{\mu^2}{U} + \frac{U}{2}|B_\eta|^2 (\lambda^2 + \lambda) 
\\&+ U\left(\frac{\mu}{U} - |B_\eta| \lambda\right) |B_\eta| (\lambda - \sqrt{\lambda^2 + \lambda}) + \frac{U}{2} |B_\eta|^2 \lambda^2 + \frac{\mu^2}{2U}
\\& = \lambda \int_{B_\eta}\eps(p)dp - \frac{\mu^2}{2U} 
\\&+U|B_\eta|^2\left(\frac{2\lambda^2 + \lambda}{2} + \frac{\lambda^2}{\lambda + \sqrt{\lambda^2 + \lambda}} \right) - \mu|B_\eta| \frac{\lambda}{\lambda + \sqrt{\lambda^2 + \lambda}}
\end{align*}
We can further estimate the kinetic term by using inequality $2 - 2\cos x \le x^2$ and obtain
$$\lambda \int_{B_\eta}\eps(p)dp \le \lambda C |B_\eta|^{\frac{5}{3}},$$
where $C$ is a constant independent of any of the parameters. Now choose $\lambda$ large enough so that
$$- \mu|B_\eta| \frac{\lambda}{\lambda + \sqrt{\lambda^2 + \lambda}} \le -\frac{\mu}{3} |B_\eta|,$$
which is possible as the above $\lambda$-dependent expression converges to $\frac12$. Next choose $\eta$ small enough so that
$$\lambda \int_{B_\eta}\eps(p)dp +U|B_\eta|^2\left(\frac{2\lambda^2 + \lambda}{2} + \frac{\lambda^2}{\lambda + \sqrt{\lambda^2 + \lambda}} \right) - \mu|B_\eta| \frac{\lambda}{\lambda + \sqrt{\lambda^2 + \lambda}} < 0.$$
This choice is also possible as the dependence in $|B_\eta|$ of the above positive terms is of higher order than the dependence of the negative term, i.e. the negative term converges to zero slower than the positive terms. It follows that we have a strict inequality
\begin{equation}
\label{strict_inf}
\inf_\D \F(\gamma, \alpha, \rho_0) < -\frac{\mu^2}{2U}. 
\end{equation}
This holds also for the restricted problem for sufficiently large $\kappa$ (such that $\lambda$ chosen before is smaller than $\kappa$), so we can write
\begin{align*}
    -\frac{\mu^2}{2U} &> \F(\gamma_\kappa, \alpha_\kappa, \rho_{0,\kappa})
    \\&\ge U\rho_{0,\kappa} \int_{\T^3} (\gamma_\kappa(p) + \alpha_\kappa(p)) dp -\mu(\rho_{\gamma_\kappa} + \rho_{0,\kappa})  + \frac{U}{2}(\rho_{\gamma_\kappa} + \rho_{0,\kappa})^2
    \\& \ge  U\rho_{0,\kappa} \int_{\T^3} (\gamma_\kappa(p) + \alpha_\kappa(p)) dp - \frac{\mu^2}{2U}
\end{align*}
Since $U > 0$ it follows that
$$\rho_{0,\kappa} \int_{\T^3} (\gamma_\kappa(p) + \alpha_\kappa(p)) dp < 0,$$
which implies
$$\rho_{0,\kappa} > 0$$
and
$$\int_{\T^3} \alpha_\kappa(p)dp < \int_{\T^3} (\gamma_\kappa(p) + \alpha_\kappa(p))dp < 0.$$
From the first inequality we can deduce that variational derivative of $\F$ with respect to $\rho_0$ must be zero:
$$\frac{\partial \F}{\partial \rho_0} = -\mu + U\int_{\T^3} \alpha_\kappa(q)dq + 2U\int_{\T^3} \gamma_\kappa(q)dq + U\rho_{0,\kappa} = 0,$$
so
$$-\mu + 2U\int_{\T} \gamma_\kappa(q)dq + U\rho_{0,\kappa} = -U\int_{\T^3} \alpha_\kappa(q)dq.$$
Using this we can rewrite variational derivative of $\F$ with respect to $\gamma$ as
\begin{equation}
\label{negative_alpha_zero}
\begin{split}
\frac{\partial \F}{\partial \gamma}  &= \eps(p) - \mu + 2U\int_{\T^3} \gamma_\kappa(q)dq + 2U\rho_0
\\&= \eps(p)  -U\int_{\T^3} \alpha_\kappa(q)dq + U\rho_{0,\kappa} > 0
\end{split}
\end{equation}
This implies that (up to the zero measure set) functions $\gamma_\kappa$ and $\alpha_\kappa$ have to be related through strict equality
\begin{equation}
\label{pure_state}
\alpha_\kappa^2(p) = \gamma_\kappa(p)(\gamma_\kappa(p)+1)
\end{equation}
Indeed, if this equality does not hold on a positive measure set then it is possible to decrease the energy by decreasing $\gamma$. This contradicts the fact of the original state being a minimizer.
\end{proof}

\begin{remark}
\label{negative_alpha}
In the above proof we obtained the fact that $\int \alpha_\kappa$ is negative, so there exists a positive measure set $M \subset \T^3$ such that for $p \in M$ we have
$$\alpha_\kappa(p) < 0.$$
From inequality \eqref{negative_alpha_zero} evaluated at $p = 0$ and continuity of $\eps(p)$ it also follows that we have strict inequality
$$\int_{\T^3} \alpha_\kappa(q) dq < \rho_{0,\kappa}.$$
\end{remark}
\begin{remark}\label{rem:positiverhoT0}
As the strict inequality \eqref{strict_inf} is valid for the unrestricted problem, from computations in the above proposition it also follows that any minimizer of the unrestricted problem (if it exists) also has to satisfy $\rho_0 > 0$ and $\int \alpha < 0$.
\end{remark}
\begin{lemma}
\label{gamma_reduce}
There exist a constant $K$ independent of $\kappa$ such that for every $\kappa$ we have
$$\gamma_\kappa(p) \le K$$
for almost all $p \in \T^3$.
\end{lemma}
\begin{proof}

\noindent \textit{Step 1.} Notice that by Proposition \ref{purity} for almost all $p$ for minimizing $\gamma_\kappa$ and $\alpha_\kappa$ we have
\begin{equation}
\label{gamma_alpha}
\gamma_\kappa(p) = \sqrt{\alpha_\kappa^2(p) + \frac14} - \frac12,  
\end{equation}
so we can define $\gamma_\kappa$ in terms of $\alpha_\kappa$:
$$\gamma_\kappa = \Phi(\alpha_\kappa)$$
where
$$\Phi(\alpha) = \sqrt{\alpha^2 + \frac14} - \frac12.$$
We will consider variation of $\F$ along the curve (in the function space) with $\gamma = \Phi(\alpha)$.

First, similarly as before, we will introduce a new notation for variational derivatives of the BBH functional $\F$:
\begin{align*}
    &A_\kappa(p) = \frac{\partial \F}{\partial \gamma} = \eps(p) - \mu + 2U\int_{\T^3} \gamma_\kappa(q)dq + 2U\rho_{0,\kappa},\\
    &B_\kappa = \frac{\partial \F}{\partial \alpha} = U\int_{\T^3} \alpha_\kappa(q)dq + U\rho_{0,\kappa}.
\end{align*}
Note that $B_\kappa$ is constant with respect to $p$. We already know that $\rho_{0,\kappa} > 0$ (cf. Proposition \ref{purity}) and therefore variational derivative of $\F$ with respect to $\rho_0$ is zero:
$$\frac{\partial \F}{\partial \rho_0} = -\mu + U\int_{\T^3} \alpha_\kappa(q)dq + 2U\int_{\T^3} \gamma_\kappa(q)dq + U\rho_{0,\kappa} = 0.$$
Using this we can rewrite the expression for $A_\kappa(p)$ as
$$A_\kappa(p) = \eps(p) - U\int_{\T^3} \alpha_\kappa(q)dq + U\rho_{0,\kappa}>0$$
because $\int \alpha_\kappa$ is negative.

Now, we can express the variation of the functional along the aforementioned curve with respect to $\alpha$ as
$$\frac{\partial}{\partial \alpha}\F\left(\Phi(\alpha_\kappa), \alpha_\kappa, \rho_{0,\kappa}\right) = A_\kappa(p) \frac{\partial \Phi}{\partial \alpha} + B_\kappa.$$
Computing the derivative of $\Phi$ with respect to $\alpha$ we get
$$\frac{\partial \Phi}{\partial \alpha} = \frac{\alpha}{\sqrt{\alpha^2 + \frac14}}.$$
As $\alpha_\kappa$ is a minimizer lying on the boundary of the domain (in the sense of \eqref{eq:pureqfT0}), the following inequalities hold on the set of full measure:
$$\frac{\partial}{\partial \alpha}\F(\Phi(\alpha_\kappa), \alpha_\kappa, \rho_{0,\kappa}) =  A_\kappa(p) \frac{\alpha_\kappa(p)}{\sqrt{\alpha_\kappa^2(p) + \frac14}} + B_\kappa =\left\{\begin{array}{l}
     = 0 \text{ if } \alpha_\kappa(p) = 0, \\
     \ge 0\text{ if } \alpha_\kappa(p) < 0,\\
     \le 0\text{ if } \alpha_\kappa(p) > 0. 
\end{array}\right.$$
By Remark \ref{negative_alpha} we know that $\alpha_\kappa(p) < 0$ on some positive measure set. For $p$'s from this set we have
\begin{equation}
\label{positive_Bk}
B_k \ge -A_\kappa(p)\frac{\alpha_\kappa(p)}{\sqrt{\alpha_\kappa^2 + \frac14}} > 0,
\end{equation}
so $B_\kappa$ is a positive constant. It follows that $\alpha_\kappa(p) \le 0$ almost everywhere, because for $p$'s such that $\alpha_\kappa(p) > 0$ inequality 
$$ A_\kappa(p) \frac{\alpha_\kappa(p)}{\sqrt{\alpha_\kappa^2 + \frac14}} + B_\kappa \le 0$$
cannot be fulfilled (left hand side is strictly positive). Knowing this we can write
$$\alpha_\kappa(p) = -|\alpha_\kappa(p)|$$
and obtain the inequality
$$-A_\kappa(p)\frac{|\alpha_\kappa|}{\sqrt{\alpha_\kappa^2 + \frac14}} + B_\kappa \ge 0$$
which is equivalent to
$$B_\kappa^2 \ge A_\kappa^2(p)\frac{\alpha_\kappa^2(p)}{\alpha_\kappa^2(p) + \frac14},$$
and after rewriting
\begin{equation}
\label{alpha_est}
\alpha_\kappa^2(p) \le \frac{B^2_\kappa}{4(A^2_\kappa(p) - B^2_\kappa)}.
\end{equation}

\noindent \textit{Step 2.} Our goal is to estimate $B^2_\kappa$ from above and $A^2_\kappa(p) - B^2_\kappa$ from below by constants independent of $\kappa$. Recall that $(\gamma_\kappa, \alpha_\kappa, \rho_{0,\kappa})$ is a minimizing sequence for a functional $\F$ and by Corollary \ref{bounded_seq}
$$\rho_{\gamma_\kappa} + \rho_{0,\kappa} < C$$
for some constant $C$ independent of $\kappa$ (dependent on $U$ and $\mu$). This implies that there exist subsequences such that
$$\rho_{\gamma_\kappa} \to \tilde{\rho_\gamma}, \; \rho_{0,\kappa} \to \tilde{\rho}_0.$$
From now on we will only work with those subsequences. Next observe that using the fact $|\alpha(p)| \le \sqrt2\max\{\gamma^{1/2}, \gamma\}$
\begin{align*}
  \left|\int_{\T^3} \alpha_\kappa(p)dp\right| &\le \int_{\T^3} |\alpha(p)|dp \le \sqrt2(1 + \|\gamma_\kappa\|_1) \le C
\end{align*}
so there also exist a convergent subsequence of the sequence $\int \alpha_\kappa(p)dp \to \omega$ for some $\omega \le 0$.

We will once again use \eqref{strict_inf}, that is the facts
$$\inf_\D \F(\gamma, \alpha, \rho_0) < -\frac{\mu^2}{2U},$$
and
$$\inf_{\D_\kappa} \F(\gamma, \alpha, \rho_0) < -\frac{\mu^2}{2U},$$
for sufficiently large $\kappa$. As minimizers of the restricted problems form a minimizing sequence of the unrestricted problem, we have
$$\lim_{\kappa \to \infty}\F(\gamma_\kappa, \alpha_\kappa, \rho_{0,\kappa}) < -\frac{\mu^2}{2U}.$$
Proceeding similarly as in the previous section we can estimate
\begin{align*}
    -\frac{\mu^2}{2U} &> \lim_{\kappa \to \infty}\F(\gamma_\kappa, \alpha_\kappa, \rho_{\kappa})
    \\&\ge\lim_{\kappa \to \infty} \left[U\rho_{0,\kappa} \int_{\T^3} (\gamma_\kappa(p) + \alpha_\kappa(p)) dp -\mu(\rho_{\gamma_\kappa} + \rho_{0,\kappa})  + \frac{U}{2}(\rho_{\gamma_\kappa} + \rho_{0,\kappa})^2 \right]
    \\& = U\tilde{\rho}_{0}(\tilde{\rho}_\gamma + \omega) -\mu(\tilde{\rho}_{\gamma} + \tilde{\rho}_{0})  + \frac{U}{2}(\tilde{\rho}_{\gamma} + \tilde{\rho}_{0})^2
    \\&\ge U\tilde{\rho}_{0}(\tilde{\rho}_\gamma + \omega) -\frac{\mu^2}{2U}.
\end{align*}
It follows that we have a strict inequality
$$U\tilde{\rho}_{0}(\tilde{\rho}_\gamma + \omega) < 0.$$
Since $U>0$ and $\tilde{\rho}_{0}$ is a limit of positive sequence so it cannot be negative, we obtain strict inequalities
\begin{align*}
    &\tilde{\rho}_{0} > 0,\\
    &\tilde{\rho}_\gamma + \omega < 0.
\end{align*}
The second inequality also implies that
$$\omega < 0.$$
We deduce that there exist $\eta > 0$ such that for sufficiently large $\kappa$ we have
$$\rho_{0,\kappa} > \eta$$
and
$$\int_{\T^3}\alpha_\kappa(p)dp < -\eta.$$

\noindent \textit{Step 3.} Now we are ready to prove the desired bounds. First we will prove the upper bound for $B_\kappa$. We can see that
$$B_\kappa = U\int_{\T^3} \alpha_\kappa(q)dq + U\rho_{0,\kappa} \le U\rho_{0,\kappa} \le C$$
because $\rho_{0,\kappa}$ is bounded as a convergent sequence, so indeed $B_\kappa$ is bounded from above. 

To obtain the second bound we write
$$A^2_\kappa(p) - B^2_\kappa = (A_\kappa(p) + B_\kappa)(A_\kappa(p) - B_\kappa)$$
and estimate each term. By the previous results we have
$$A_\kappa(p) + B_\kappa = \eps(p) + 2U\rho_{0,\kappa} \ge 2U\eta$$
and
$$A_\kappa(p) - B_\kappa = \eps(p) - 2U\int_{\T^3} \alpha_\kappa(p) dp \ge 2U\eta.$$
Finally we get
$$A^2_\kappa(p) - B^2_\kappa \ge 4U^2\eta^2 > 0.$$
Using obtained estimations in inequality \eqref{alpha_est} we get
$$\alpha_\kappa^2(p) \le \frac{C^2}{16U^2\eta^2} =: K$$
Recalling the relation between $\gamma_\kappa$ and $\alpha_\kappa$ we can also write
$$\gamma_\kappa(p) \le \gamma_\kappa(p) (\gamma_\kappa(p) +1) = \alpha^2_\kappa(p) \le K.$$
\end{proof}

\begin{remark}
\label{ineq_T=0}
From the strict inequality \eqref{positive_Bk} and negativity of $\alpha_\kappa$ it follows that we have the strict inequality
$$\rho_{0,\kappa} > \|\alpha_\kappa\|_1.$$ 
\end{remark}
 
We are ready to prove the main result.
\begin{prop}
\label{T=0ex}
For $T=0$ there exist a minimizer of the unrestricted problem. 
\end{prop}

\begin{proof}
We will use the fact that $(\gamma_\kappa, \alpha_\kappa, \rho_{0,\kappa})$ forms a minimizing sequence and, with the proper selection of a subsequence, all previous estimates hold. Let $\kappa_0$ be the smallest $\kappa$ such that $\kappa_0 > K$. Observe that for $\kappa \ge \kappa_0$ we have $(\gamma_\kappa, \alpha_\kappa, \rho_{0,\kappa}) \in \D_{\kappa_0}$. Next we note that
$\F(\gamma_\kappa, \alpha_\kappa, \rho_{0,\kappa})$ is decreasing. Then we have
\begin{align*}
    \inf_\D \F(\gamma, \alpha, \rho_0) &= \lim_{\kappa \to \infty}\F(\gamma_\kappa, \alpha_\kappa, \rho_{0,\kappa})
    \\&= \inf_\kappa\F(\gamma_\kappa, \alpha_\kappa, \rho_{0,\kappa})
    \\&\ge  \inf_{\D_{\kappa_0}}\F(\gamma, \alpha, \rho_0)
    \\&= \F(\gamma_{\kappa_0}, \alpha_{\kappa_0}, \rho_{0,\kappa_0})
    \\&\ge \inf_\D \F(\gamma, \alpha, \rho_0),
\end{align*}
which proves that $(\gamma_{\kappa_0}, \alpha_{\kappa_0}, \rho_{0,\kappa_0})$ is the minimizer of the unrestricted problem. 
\end{proof}
\begin{remark}
Due to Remark \ref{rem:T>0proof}, Proposition \ref{T=0ex} concludes the proof of Theorem \ref{main_existence}.    
\end{remark}
 
Note that in all the proofs in this subsection ($T=0$ case) the only assumption on parameter $U$ we used is $U>0$. In particular for any $U > 0$ minimizer exists and minimizing $\rho_0$ is positive. This means that the structure of the minimizer does not depend on $U$ and hence we obtain a following
\begin{cor}
\label{QPTnot}
A quantum phase transition (i.e. transition with respect to the parameter $U$ for temperature $T=0$) does not occur.
\end{cor}

\section{The phase transition and the phase diagram}
\label{section_phase}

In this section we will give the proof of Theorem  \ref{main_structure}. Recall Euler-Lagrange inequalities: if $(\gamma, \alpha, \rho_0)$ is the minimizer of the BBH functional $\F$ then
\begin{equation}
\begin{split}
\label{E_L_full}
\frac{\partial \F}{\partial \gamma}  &= \eps(p) - \mu + 2U\int_{\T^3} \gamma(q)dq + 2U\rho_{0} - T \frac{\gamma + \frac12}{\beta}\ln\frac{\beta + \frac12}{\beta - \frac12} = 0\\
\frac{\partial \F}{\partial \alpha} &= U\int_{\T^3} \alpha(q)dq + U\rho_{0} + T\frac{\alpha}{\beta}\ln\frac{\beta + \frac12}{\beta - \frac12} = 0\\
\frac{\partial \F}{\partial \rho_0} &= -\mu + U\int_{\T^3} \alpha(q)dq + 2U\int_{\T^3} \gamma(q)dq + U\rho_{0} =\left\{\begin{array}{l}
     = 0 \text{ if } \rho_{0} > 0 \\
     \ge 0\text{ if } \rho_{0} = 0.
\end{array}\right. 
\end{split}
\end{equation}
In Proposition \ref{rest_stucture} we characterized the structure of minimizers of the restricted problems for $T>0$. When concerning the unrestricted problem, we can repeat the proof to once again obtain such a characterization. For completeness, we will formulate the following proposition. By Remark \ref{ineq_T=0}, this claim is also valid for $T = 0$.
\begin{prop}
\label{min_struct}
Let $(\gamma, \alpha, \rho_{0})$ be the minimizer of the functional $\F$. Then
$$\alpha \equiv 0 \Longleftrightarrow \rho_{0} = 0.$$
If $\rho_{0} > 0$ then $\alpha$ is negative almost everywhere and we have strict inequality
$$\|\alpha\|_1 < \rho_{0}.$$
\end{prop}
From the above relations we can deduce the following corollary that will be useful later.
\begin{cor}
\label{L1_ineq}
If a minimizing $\rho_0$ is zero, then for a minimizing $\gamma$ we have
$$\|\gamma\|_1 \ge \frac{\mu}{2U}.$$
If a minimizing $\rho_0$ is non-zero, then for a minimizing $\gamma$ we have
$$\|\gamma\|_1 \le \frac{\mu}{2U}.$$
Furthermore, for $\mu \le 0$ it is impossible for the minimizer to have $\rho_0 > 0$.
\end{cor}
\begin{proof}
It follows directly from the $\rho_0$ derivative condition in \eqref{E_L_full} and Proposition \ref{min_struct}.
\end{proof}

Let us collect the existing results.
\begin{itemize}
    \item It follows from Proposition \ref{T=0ex} and Remark \ref{rem:positiverhoT0} that for $T = 0$ and $\mu > 0$ every minimizer has $\rho_0 > 0$. When $T=0$ and  $\mu \le 0$ then the vacuum is a minimizing state. This proves $(1)$ and $(2)$ of Theorem \ref{main_structure}.
    \item The last statement in Corollary \ref{L1_ineq} implies that for $T > 0$ and $\mu \le 0$ any minimizer is an insulator, i.e. it has $\rho_0 = 0, \alpha \equiv 0$. This proves $(3)$ of  Theorem \ref{main_structure}.
\end{itemize}

To finish the proof of Theorem \ref{main_structure} it remains to prove $(4)$ and $(5)$. This is the content of the two following Lemmata. 

\begin{lemma}
\label{known_area1}
For parameters satisfying:
$$T> 0, \quad \mu > 0, \quad J(T,0) < \frac{\mu}{2U}$$
every minimizer has a non-zero $\rho_0$.
\end{lemma}

\begin{proof}
Suppose minimizing $\rho_0$ is zero. We will show that this will lead to a contradiction. By Proposition \ref{rho_0=0_existence} and preceding proofs we know that in this case the minimizing $\gamma$ is given by the formula
$$\gamma(p) = \left(e^{\frac1{T}\left(\eps(p) - \mu + 2UG\right)} - 1\right)^{-1},$$
where $G \ge \frac{\mu}{2U}$ is the solution to the equation
\begin{equation}
\label{G_eq}
G = \int_{\T^3}\left(e^{\frac1{T}\left(\eps(p) - \mu + 2UG\right)} - 1\right)^{-1}.  
\end{equation}
We will show that with the chosen values of the parameters this equation cannot be fulfilled. Note that the right hand side of the above equality is non-increasing with respect to $G$. Its maximal value (obtained for $G = \frac{\mu}{2U}$) is
$$\int_{\T^3}\left(e^{\frac{\eps(p)}{T}}-1\right)^{-1} dp = J(T,0).$$
By our assumption, this value is lower than the minimal value of the left-hand side of \eqref{G_eq}, equal to $\frac{\mu}{2U}$. This is the desired contradiction.
\end{proof}

\begin{lemma}
\label{known_area2}
For parameters satisfying:
\begin{equation}
\label{lemma_assumption}
T> 0, \quad \mu > 0, \quad J(T,4\mu+2U) > \frac{\mu}{2U}   
\end{equation}
every minimizer has $\rho_0 = 0.$
\end{lemma}
\begin{proof}
We will once again argue by contradiction. Assume that $\rho_0 > 0$. Since $(\gamma, \alpha, \rho_0)$ is a minimizer, we have an inequality
$$\F(\gamma, \alpha, \rho_0) \le \F(\gamma, 0, 0),$$
which, after dropping some positive terms and using $-S(\gamma,0) \le -S(\gamma,\alpha)$ and $\gamma(p) + \alpha(p) \ge -\frac12$
yields
\begin{equation}
\label{rho_0_ineq}
-\mu \rho_0 - \frac12U\rho_0 + \frac12U\rho_0^2 \le 0 \Longrightarrow \rho_0 \le \frac{2\mu}{U} + 1.    
\end{equation}
Now we will estimate $\gamma(p)$ using a similar argument as in the proof of Corollary \ref{K_bound} and the preceding Lemma (recall in particular equation \eqref{gamma_eq}). Using  the Euler-Lagrange equations \ref{E_L_full} we can write $\gamma(p)$ as
$$\gamma(p) = \frac12 \left[\frac{e^{\sqrt{A^2(p) - B^2}}+1}{\sqrt{A^2(p) - B^2}(e^{\sqrt{A^2(p) - B^2}}-1)}A(p) - 1\right],$$
where
\begin{align*}
    A(p) &= \frac1{T}\left(\eps(p) -\mu + 2U\int_{\T^3} \gamma(q)dq + 2U\rho_0\right),\\
    B &= \frac1{T}\left(U\int_{\T^3}\alpha(q)dq + U\rho_0\right).
\end{align*}
By direct computation we can check that the function
$$x \mapsto \frac{e^x+1}{x(e^x-1)}$$
is monotone decreasing for $x>0$, hence (recall $A(p)>0$) we get
$$\gamma(p) \ge \frac12 \left[\frac{e^{A(p)}+1}{A(p)(e^{A(p)}-1)}A(p) - 1\right] = \frac{1}{e^{A(p)}-1} = \left(e^{\frac1{T}\left(\eps(p) -\mu + 2U\int_{\T^3} \gamma(q)dq + 2U\rho_0\right)}-1\right)^{-1}.$$
Using inequality \eqref{rho_0_ineq} and the bound from Corollary \ref{L1_ineq} on $\|\gamma\|_1$, i.e.
\begin{equation}
\label{lemma_assumption2}
\|\gamma\|_1 \le \frac{\mu}{2U}  
\end{equation}
we can further estimate
$$\gamma(p) \ge \left(e^{\frac1{T}\left(\eps(p) + 4\mu + 2U\right)}-1\right)^{-1}.$$
Integrating this inequality over $\T^3$ gives us
$$\|\gamma\|_1 \ge \int_{\T^3}\left(e^{\frac1{T}\left(\eps(p) + 4\mu + 2U\right)}-1\right)^{-1}dp = J(T,4\mu+2U).$$
Once again using \eqref{lemma_assumption2} it follows that
$$J(T,4\mu+2U) \le \frac{\mu}{2U},$$
which, under the assumptions \eqref{lemma_assumption}, is a contradiction.
\end{proof}
Lemma \ref{known_area1} gives the proof of $(4)$ in Theorem \ref{main_structure} and Lemma \ref{known_area2} proves $(5)$. This completes the proof of the whole statement.

\section{Proof of the theorem in the canonical setting}
\label{canonical}
The goal of this section is to provide a proof of Theorem \ref{can_existence}. The proof will rely on results from previous sections. We will not repeat all the details, but rather stress the main differences and show how to make the  required adjustments.

\begin{proof}[Proof of Theorem \ref{can_existence}.] \textit{Step 1: Well-posedness.} By the same argument as in Proposition \ref{boundedness}, the functional $\Fc$ is bounded below, so this problem is well-posed, i.e. the infimum is finite.
\\
\\
\noindent \textit{Step 2: Double minimization.} For now we focus on the formulation \eqref{can_min}. We will once again perform a  double minimization procedure, but this time we will keep fixed both values of $\rho_0$ and $\int \gamma$. To this end, similarly as in Lemma \ref{f_cont} we define
$$\widetilde \F^{\conv}(\gamma, \alpha, \rho_0) = \Fc(\gamma, \alpha, \rho_0) + U\rho_0^2$$
and for $\lambda \ge 0$ and $\rho_0 \ge 0$
$$\widetilde f^{\conv}(\lambda,\rho_0) = \inf_{\substack{(\gamma ,\alpha) \in \D' \\ \int \gamma = \lambda}} \widetilde \F^{\mathrm{conv}}(\gamma, \alpha, \rho_0).$$
As $\widetilde \F^{\conv}$ is convex in all three variables, $\widetilde f^{\conv}$ is also convex and hence continuous on $(0,\infty)\times (0,\infty)$. Slightly adapting the argument of continuity on the boundary given in the proof of the aforementioned lemma, we can conclude that $\widetilde f^{\mathrm{conv}}$ is in fact continuous for all $\lambda, \rho_0 \ge 0$. Now we introduce 
\begin{equation}
\label{double_min_can}
 \fc(\lambda, \rho_0) = \inf_{\substack{(\gamma ,\alpha) \in \D' \\ \int \gamma = \lambda}} \Fc(\gamma, \alpha, \rho_0). 
\end{equation}
Since 
$$\fc(\lambda, \rho_0)= \widetilde f^{\mathrm{conv}}(\lambda, \rho_0) -U\rho_0^2 $$
we obtain that $\fc(\lambda, \rho_0)$ is a continuous function of both variables. 
This allows us to split the minimization problem \eqref{can_min} into two steps: first  finding $\lambda^{\min} \in [0,\rho]$ such that
\begin{equation}
\label{lambda_min}
\fc(\lambda^{\min}, \rho - \lambda^{\min}) = \inf_{0 \le \lambda \le \rho} \fc(\lambda, \rho - \lambda)  
\end{equation}
and then, for fixed $\lambda^{\min}$ and $\rho_0^{\min} = \rho - \lambda^{\min}$, finding $\gamma$ and $\alpha$ minimizing problem \eqref{double_min_can}. The existence of such $\lambda^{\min}$ follows directly from continuity of $\fc$ and compactness of $[0,\rho]$ (note that such $\lambda^{\min}$ possibly is not unique). It remains to prove the existence of proper $\gamma$ and $\alpha$.
\\
\\
\noindent \textit{Step 3: Restricted problem.} Just like before, we introduce restricted problem, that is we are interested in finding minimizer of the problem
\begin{equation}
\label{restricted_can}
\inf_{\D_{\kappa}(\rho)} \Fc(\gamma, \alpha, \rho_0),   
\end{equation}
where
$$\D_{\kappa}(\rho) = \{(\gamma, \alpha, \rho_0) \in \D(\rho) \colon \gamma(p) \le \kappa \}.$$
In analogy to \eqref{double_min_can} we also introduce
\begin{equation}
\label{double_min_can_kappa}
\fc_\kappa(\lambda, \rho_0) = \inf_{\substack{(\gamma ,\alpha) \in \D'_\kappa \\ \int \gamma = \lambda}} \Fc(\gamma, \alpha, \rho_0).    
\end{equation}
% with
% $$\D'_\kappa(\rho)  = \left\{(\gamma, \alpha) \in \D' \colon \int_{\T^3} \gamma(p) dp \le \rho, \; \gamma(p) \le \kappa \right\}.$$
Using the same arguments as in the proof of Proposition \ref{restriced_existence} we deduce that:
\begin{itemize}
    \item for every $\lambda \in [0,\rho]$, for sufficiently large $\kappa$ ($\kappa > \rho$ is sufficient) there exist $(\gamma_\kappa, \alpha_\kappa)$ minimizing the problem \eqref{double_min_can_kappa}, that is
    $$\fc_\kappa(\lambda, \rho - \lambda) = \Fc(\gamma_\kappa, \alpha_\kappa, \rho - \lambda);$$
    \item in particular, for some $\lambda_\kappa$ we have
    $$\fc(\lambda_\kappa, \rho - \lambda_\kappa) = \inf_{\D_\kappa(\rho)} \Fc(\gamma, \alpha, \rho_0) = \Fc(\gamma_\kappa, \alpha_\kappa, \rho - \lambda_\kappa).$$
\end{itemize}
Next, just like in Proposition \ref{minimizing_seq} and Proposition \ref{rho_0_conv}, we can prove that $(\gamma_\kappa, \alpha_\kappa, \rho_{0,\kappa})$ with $\rho_{0,\kappa} := \rho -\lambda_\kappa$ is a minimizing sequence of \eqref{can_min}, i.e.
$$\lim_{\kappa \to \infty} \Fc(\gamma_\kappa, \alpha_\kappa, \rho_{0,\kappa}) = \inf_{\D(\rho)} \Fc(\gamma, \alpha, \rho_0)$$
and, up to a subsequence, we have
\begin{equation}
\label{lambda_kappa_conv}
\lim_{\kappa \to \infty} \lambda_\kappa = \lambda^{\min},  
\end{equation}
where $\lambda^{\min}$ is a (one of, in case of non-uniqueness) minimizer of \eqref{lambda_min}.
\\
\\
\noindent \textit{Step 4: Reformulation and separation of cases.} From this point onwards we will focus on the formulation \eqref{can_min2} and its restricted version: finding minimizer of the problem
\begin{equation}
\label{restricted_can2}
\inf_{\D'_\kappa(\rho)} \Fcr(\gamma, \alpha), 
\end{equation}
where $\D'_\kappa(\rho)$ was defined in the previous step. This is a minimization problem equivalent to \eqref{restricted_can}. The restricted minimizers $(\gamma_\kappa ,\alpha_\kappa)$ of \eqref{restricted_can} obtained above are also minimizers of \eqref{restricted_can2}, so we can deduce some of their properties using the variational derivatives
\begin{equation}
\label{variational_derivatives_can}
\begin{split}
\frac{\partial \Fcr}{\partial \gamma}  &= \eps(p) - U\int_{\T^3} \gamma(q)dq + U\rho - U\int_{\T^3} \alpha(q)dq  - T \frac{\gamma + \frac12}{\beta}\ln\frac{\beta + \frac12}{\beta - \frac12},\\
\frac{\partial \Fcr}{\partial \alpha} &= U\int_{\T^3} \alpha(q)dq + U\rho - U\int_{\T^3} \gamma(q)dq + T\frac{\alpha}{\beta}\ln\frac{\beta + \frac12}{\beta - \frac12}.
\end{split}
\end{equation}
In order to do so, similarly as in the grand canonical version of the problem, we will split our considerations into two cases depending on the temperature $T$ -- the first one is $T>0$ and the second one is $T = 0$. When $T>0$ we will further divide the analysis into two instances: when  $\lambda^{\min} < \rho$ (and, as a consequence, $\rho_0^{\min} > 0$) or $\lambda^{\min} = \rho$ (no condensation).
\\
\\
\noindent \textit{Case 1: Positive temperature, existence of condensation.} For now on we assume that $T > 0$ and $\lambda^{\min} < \rho$. The second condition implies that for sufficiently large $\kappa$ we have strict inequality
$$\lambda_\kappa = \int_{\T^3}\gamma_\kappa(p) dp < \rho.$$
Similarly as in Lemma \ref{bound_below} and Lemma \ref{alpha_ineq} we can then prove that 
\begin{equation} \label{eq:gamma_lower_bound_can}
 \gamma_\kappa \ge c > 0,\qquad \text{and} \qquad    \alpha^2_\kappa < \gamma_\kappa(\gamma_\kappa + 1)
\end{equation}
  for almost all $p \in \T^3$. Indeed, the essential steps in the first of those proofs were the following:
\begin{itemize}
    \item On the set $\{\gamma_\kappa(p) < \kappa\}$ we had an inequality
    $$\frac{\partial \F}{\partial \gamma}(p) \le \eps(p) + C - T \ln \frac1{\beta_\kappa - \frac12}.$$
    Here it is also true as
    \begin{equation}
    \label{can_gamma_est}
    - U\int_{\T^3} \gamma_\kappa(q)dq + U\rho - U\int_{\T^3} \alpha_\kappa(q)dq \le U\rho +  U\left|\int_{\T^3} \alpha_\kappa(q)dq\right| \le C(\rho),
    \end{equation}
    which follows from $\|\alpha\|_1 \le C(1 + \|\gamma\|_1) \le C(1 + \rho)$.

    \item Increasing $\gamma_\kappa$ on the certain set would lower the energy. Here, as $\int \gamma_\kappa < \rho$, such increment of $\gamma_\kappa$ is possible without violating the density constraint.
\end{itemize}
It follows that those proofs can be adapted to the canonical setting and therefore the inequalities \eqref{eq:gamma_lower_bound_can} hold true. Note that this adaptation potentially could not be possible if $\lambda^{\min} = \rho$.

With this knowledge, the Euler-Lagrange derivatives \eqref{variational_derivatives_can} for restricted problem have to satisfy  analogous conditions to the ones stated in Corollary \ref{EL+}, that is
\begin{equation}
\label{variational_derivatives_can_rest}
\begin{split}
\frac{\partial \Fcr}{\partial \gamma}  &= \eps(p) - U\int_{\T^3} \gamma_\kappa(q)dq + U\rho - U\int_{\T^3} \alpha_\kappa(q)dq  - T \frac{\gamma_\kappa + \frac12}{\beta_\kappa}\ln\frac{\beta_\kappa + \frac12}{\beta_\kappa - \frac12} = \left\{ 
\begin{array}{l}
     = 0 \text{ if } \gamma_\kappa(p) < \kappa \\
     \le 0\text{ if } \gamma_\kappa(p) = \kappa
\end{array}\right.,\\
\frac{\partial \Fcr}{\partial \alpha} &= U\int_{\T^3} \alpha_\kappa(q)dq + U\rho - U\int_{\T^3} \gamma_\kappa(q)dq + T\frac{\alpha_\kappa}{\beta_\kappa}\ln\frac{\beta_\kappa + \frac12}{\beta_\kappa - \frac12} = 0.
\end{split}
\end{equation}
Furthermore, using the same argument as in Proposition \ref{rest_stucture}, we can deduce that $\alpha_\kappa$ is either identically zero or everywhere negative function and
\begin{equation}
\label{can_alpha_est}
\|\alpha_\kappa\|_1 < \rho - \int_{\T^3} \gamma_\kappa(p)dp,
\end{equation}
that is the norm of $\alpha$ is bounded by the condensate density $\rho_{0,\kappa}$.

For $p$'s such that $\gamma_\kappa(p) < \kappa$ we can express $\gamma_\kappa$ just like in \eqref{gamma_eq} using $A_\kappa(p)$ and $B_\kappa$ defined analogously as in \eqref{AB_intro}, namely
\begin{equation*}
\begin{split}
   A_\kappa(p) &= \frac1{T}\left(\eps(p) - U\int_{\T^3} \gamma_\kappa(q)dq + U\rho - U\int_{\T^3} \alpha_\kappa(q)dq\right),\\
    B_\kappa &= \frac1{T}\left(U\int_{\T^3} \alpha_\kappa(q)dq + U\rho - U\int_{\T^3} \gamma_\kappa(q)dq\right) 
\end{split}
\end{equation*}
and
\begin{equation}
\label{gamma_kappa2}
\gamma_\kappa(p) = \frac12 \left[\frac{e^{\sqrt{A^2_\kappa(p) - B^2_\kappa}}+1}{\sqrt{A^2_\kappa(p) - B^2_\kappa}(e^{\sqrt{A^2_\kappa(p) - B^2_\kappa}}-1)}A_\kappa(p) - 1\right]. 
\end{equation}
We would like to justify that statement in Lemma \ref{Cc} also holds in this case, that is
$$A_\kappa(p) \le C$$
$$\sqrt{A^2_\kappa(p) - B_\kappa^2} \ge c$$
for some constants $C, c > 0$. Here we have $A_\kappa(p) \le C(\rho)$ by \eqref{can_gamma_est} and, by \eqref{can_alpha_est} and density constraint, we get
$$T(A_\kappa(p) - B_\kappa) = \eps(p) - 2U\int_{\T^3}\alpha_\kappa(q)dq > 2U\|\alpha_\kappa\|_1$$
and
$$T(A_\kappa(p) + B_\kappa) = \eps(p) + 2U\rho - 2U\int_{\T^3}\gamma_\kappa(q)dq \ge 2U\rho_{0,\kappa}.$$
These are exactly the same estimates used in the proof of the aforementioned proposition. It follows that this proof is also valid for the canonical setting.

All remaining propositions in the grand canonical setting following Proposition \ref{Cc} were proved based on this proposition alone and the fact that increasing $\gamma_\kappa$ on certain sets leads to a contradiction. Since we already know that it is possible to adapt this proposition to the canonical setting (meaning without violating the density constraint), it follows that all subsequent statements remain true in this case. This ends this part of the proof concerning the existence of canonical minimizers.
\\
\\
\noindent \textit{Case 2: Positive temperature, no condensation.} In this subsection we assume that the infimum of \eqref{can_min} is obtained for $\lambda^{\min} = \rho$ and therefore $\rho_0^{\min} = 0$. First we will prove that if
\begin{equation}
\label{J(T)_rho}
J(T,0) < \rho,  
\end{equation}
then this situation is not possible, i.e. if the temperature is sufficiently low then the minimizing $\rho_0^{\min}$ cannot be zero.

We will argue by contradiction. Suppose that this is the case, then, by \eqref{lambda_kappa_conv}, we have $\lambda_\kappa \to \rho$. Passing to a subsequence, we can assume that either $\lambda_\kappa < \rho$ or $\lambda_\kappa = \rho$ for all $\kappa$. 

In the first case ($\lambda_k < \rho$) we can express $\gamma_\kappa$ like in \eqref{gamma_kappa2}. Then, using \eqref{can_alpha_est}, we can deduce
$$\lim_{\kappa \to \infty} B_\kappa = 0$$
and
$$\lim_{\kappa \to \infty} A_\kappa(p) = \eps(p).$$
Therefore, for $p \ne 0$, the sequence $\gamma_\kappa$ converges pointwise:
$$\lim_{\kappa \to \infty} \gamma_\kappa(p) = \left(e^{\eps(p)/T} - 1\right)^{-1}.$$
By Fatou lemma we have
$$\rho = \lim_{\kappa \to \infty} \lambda_\kappa = \lim_{\kappa \to \infty} \int_{\T^3} \gamma_\kappa(p) dp \ge \int_{\T^3} \left(e^{\eps(p)/T} - 1\right)^{-1} dp = J(T,0),$$
which contradicts \eqref{J(T)_rho}.

In the second case ($\lambda_\kappa = \rho$), by the monotonicity of the entropy, for any $(\gamma, \alpha) \in \D'$ we have
$$\Fcr(\gamma, 0, \rho_0= 0) \le \Fcr(\gamma, \alpha, \rho_0= 0)$$
and therefore $\alpha_\kappa \equiv 0$. Furthermore, as $\gamma_\kappa$ is a minimizer of the functional under the constraint $\int \gamma \le \rho$, there exists a Lagrange multiplier $\nu_\kappa \le 0$  such that
\begin{equation}
\label{necessary}
\frac{\partial}{\partial \gamma}\left(\Fcr - \nu_\kappa \int \gamma\right) = \eps(p) - T\ln\left(1 + \frac1{\gamma_\kappa(p)}\right) - \nu_\kappa = \left\{ 
\begin{array}{l}
     = 0 \text{ if } \gamma_\kappa(p) < \kappa \\
     \le 0\text{ if } \gamma_\kappa(p) = \kappa.
\end{array}\right.    
\end{equation}

Note that the Lagrange multiplier did not appear (that is $\nu_\kappa = 0$) in the previous case as the condition $\lambda^{\min} < \rho$ meant that the constraint was not active. Solving this relation for $\gamma_\kappa$ we get
$$\gamma_\kappa(p) = \min \left\{\kappa, \left(e^{\frac1{T}(\eps(p) - \nu_\kappa)} - 1\right)^{-1}\right\}.$$
Passing to a subsequence we can further assume that the sequence $\nu_\kappa$ is convergent to some $\nu_0 \le 0$ (possibly $\nu_0 = - \infty$). Using this fact we can repeat the argument used in the proofs of Lemma \ref{pointwise} and Lemma \ref{zero_difference} -- we can show, with a proper selection of a subsequence, that $\gamma_\kappa$ converges pointwise to
$$\tilde \gamma(p) = \left(e^{\frac1{T}(\eps(p) - \nu_0)}-1\right)^{-1}$$
and integrals of $\gamma_\kappa$ (that is $\lambda_\kappa$) converge to the integral of the above function, as otherwise we could strictly decrease the energy by moving mass to the condensate, which contradicts the fact $\lambda^{\min} = \rho$. It follows that the following equality needs to hold true:
\begin{equation}
\label{mu_0_int}
\int_{\T^3} \left(e^{\frac1{T}(\eps(p) - \nu_0)}-1\right)^{-1} dp = \rho.    
\end{equation}
Treating the left hand side as a function of $\nu_0 \le 0$, its highest value is obtained for $\nu_0 = 0$ and equals 
$$\int_{\T^3} \left(e^{\eps(p)/T}-1\right)^{-1} dp = J(T,0).$$
This is a contradiction with \eqref{J(T)_rho} as the right hand side of \eqref{mu_0_int} is strictly larger than that.

Now we shall prove that if we assume the reverse inequality in \eqref{J(T)_rho}, that is
$$J(T,0) \ge \rho,$$
then indeed there exist minimizer of \eqref{can_min2} with $\int \gamma = \rho$. This follows from considerations similar to those above, as \eqref{mu_0_int} does not yield a contradiction. Furthermore, it can be shown that as $\gamma_\kappa$ converges in $L^1(\T^3)$ then the limiting function is a desired minimizer.
\\
\\
\noindent \textit{Case 3: Zero temperature.} First we will show that for $T = 0$ we have a strict inequality
$$\inf_{\D'(\rho)}\Fcr(\gamma, \alpha) < \frac{U}{2}\rho^2,$$
that is the ground state energy is strictly smaller that the energy of the pure condensation state. This follows from performing analogous computations as in Proposition \ref{purity}, i.e. considering trial states
$$\gamma(p) = \lambda \mathbbm{1}_{B_\eta}(p), \; \; \alpha(p) = -\sqrt{\gamma(p)(\gamma(p)+1)},$$
where $B_\eta$ is a ball centered at $p=0$ with radius $\eta$ and both parameters $\lambda$ and $\eta$ are chosen in such a way that this state belongs to the domain $\D'(\rho)$. Evaluating  the canonical BBH functional \eqref{BBHcan} in this state we get
\begin{align*}
\Fcr(\gamma, \alpha) &= \lambda\int_{B_\eta}\eps(p)dp + \frac{U}{2}\rho^2
\\&+ \frac{U}{2}|B_\eta|^2\left(2\lambda^2 + \lambda \right)
\\&+ U(\rho - |B_\eta|\lambda)(-|B_\eta|\sqrt{\lambda^2 + \lambda} + |B_\eta|\lambda)
\\& < \frac{U}{2}\rho^2
\end{align*}
for an appropriate choice of $\lambda$ and $\eta$. This inequality is also true for minimizers of the restricted problem (for sufficiently large $\kappa$). As the term $(\rho - \int \gamma)(\int \gamma + \int \alpha)$ is the only one in the functional that can possibly yield negative value, we have
\begin{equation}
\label{can_ineq}
\begin{split}
\frac{U}{2}\rho^2 &> \Fcr(\gamma_\kappa, \alpha_\kappa)
\\& \ge \frac{U}{2}\rho^2 + \left(\rho - \int_{\T^3} \gamma_\kappa(q)dq\right)\left(\int_{\T^3} \gamma_\kappa(q)dq + \int_{\T^3} \alpha_\kappa(q)dq\right)
\end{split}
\end{equation}
from which it follows that we have strict inequalities
$$\int_{\T^3} \gamma_\kappa(q)dq + \int_{\T^3} \alpha_\kappa(q) dq < 0$$
and
$$\rho - \int_{\T^3} \gamma_\kappa(q)dq > 0.$$
Analogously, as in \eqref{negative_alpha_zero}, we deduce
$$\frac{\partial \Fcr}{\partial \gamma} = \eps(p) - U\int_{\T^3} \gamma_\kappa(q)dq + U\rho - U\int_{\T^3} \alpha_\kappa(q)dq > 0,$$
so that the equality
$$\alpha_\kappa^2(p) = \gamma_\kappa(p)(\gamma_\kappa(p)+1)$$
must hold almost everywhere.

Now, just like in Lemma \ref{gamma_reduce} we repeat the procedure of reducing the $\gamma$ variable in the functional by setting $\gamma = \Phi(\alpha) = \sqrt{\alpha^2 + \frac14} - \frac12$. By repeating the steps in the proof of that proposition, we obtain the inequality
$$\alpha^2_\kappa(p) \le \frac{B_\kappa}{4(A^2_\kappa(p) -  B^2_\kappa)}.$$
We want to prove that the right hand side is bounded by the constant independent of $\kappa$ (up to a subsequence). This follows from similar reasons as in the mentioned proof -- as the sequence of integrals $\int \alpha_\kappa$ is bounded, we can pass to a subsequence such that $\int \alpha_\kappa$ converges to some $\omega \le 0$ (recall that we already know $\int \gamma_\kappa$ converges to $\lambda^{\min}$). Using this and the fact that $(\gamma_\kappa, \alpha_\kappa)$ is a minimizing sequence of the problem \eqref{can_min2} we can improve \eqref{can_ineq} and its conclusions by noting that
\begin{align*}
\frac{U}{2}\rho^2 &> \lim_{\kappa \to \infty} \Fcr(\gamma_\kappa, \alpha_\kappa)
\\&\ge \lim_{\kappa \to \infty} \left[\frac{U}{2}\rho^2 + \left(\rho - \int_{\T^3} \gamma_\kappa(q)dq\right)\left(\int_{\T^3} \gamma_\kappa(q)dq + \int_{\T^3} \alpha_\kappa(q)dq\right)\right]
\\&= \frac{U}{2}\rho^2 + \left(\rho - \lambda^{\min} \right)\left(\lambda^{\min} + \omega\right)
\end{align*}
and hence
$$\left(\rho - \lambda^{\min} \right)\left(\lambda^{\min} + \omega\right) < 0.$$
This means that
$$\rho - \lambda^{\min} > 0$$
and
$$\lambda^{\min} + \omega < 0.$$
Therefore, for sufficiently large $\kappa$ we have
\begin{equation}
\label{can_eta}
\begin{split}
&\rho - \int_{\T^3}\gamma_\kappa(q)dq > \eta > 0,
\\&\int_{\T^3} \alpha_\kappa(q) \le \int_{\T^3} \gamma_\kappa(q)dq + \int_{\T^3} \alpha_\kappa(q) dq \le -\eta < 0.
\end{split}
\end{equation}
for some $\eta > 0$.

Next, using \eqref{can_eta}, for sufficiently large $\kappa$ we have
$$A_\kappa(p) - B_\kappa = \eps(p) - 2U\int_{\T^3}\alpha_\kappa(q)dq = \eps(p) + 2U\|\alpha_\kappa\|_1 > 2U\eta$$
and
$$A_\kappa(p) + B_\kappa = \eps(p) + 2U\rho - 2U\int_{\T^3}\gamma_\kappa(q)dq > 2U\eta.$$
The remaining part of the proof is exactly the same as in the grand canonical version (i.e. Lemma  \ref{gamma_reduce} and Proposition \ref{T=0ex}). This ends the proof of the existence part in Theorem \ref{can_existence}.
\\
\\
\noindent \textit{Step 5.} To finish the proof of Theorem \ref{can_existence}
we first notice that the fact that below $T_1$ the minimizer has $\rho_0 > 0$ follows from the analysis in Step 4 (Case 2) for $T>0$ and Step 4 (Case 3) for $T=0$. This proves $(1)$. 

The proof of $(2)$ is analogous to the proof of Lemma \ref{known_area2}. Suppose that for $T \ge T_2$ (in particular $J(T,2U\rho) \ge J(T_2,2U\rho)$) there exist a minimizer $(\gamma, \alpha, \rho_0)$ with $\rho_0 > 0$. In particular, as we are working with the constrained problem, this implies
\begin{equation}
\label{int_gamma<rho}
\|\gamma\|_1 < \rho.
\end{equation}
Using the fact that in this case both variational derivatives \eqref{restricted_can2} are zero we can write, similarly as in \eqref{gamma_kappa2} that
$$\gamma(p) = \frac12 \left[\frac{e^{\sqrt{A^2(p) - B^2}}+1}{\sqrt{A^2(p) - B^2}(e^{\sqrt{A^2(p) - B^2}}-1)}A(p) - 1\right]$$
with
\begin{align*}
    A(p) &= \frac1{T}\left(\eps(p) - U\int_{\T^3} \gamma(q)dq + U\rho - U\int_{\T^3} \alpha(q)dq\right),\\
    B &= \frac1{T}\left(U\int_{\T^3} \alpha(q)dq + U\rho - U\int_{\T^3} \gamma(q)dq\right).
\end{align*}
We estimate $\gamma(p)$ pointwise in the same way as before
$$\gamma(p) \ge \left(e^{A(p)}-1\right)^{-1}$$
and by inequality \eqref{can_alpha_est} (which holds also for the unrestricted minimizer by the same argument)  we further estimate
$$A(p) \le \frac1{T}\left(\eps(p) + 2U\rho - 2U\int_{\T^3}\gamma(q)dq\right) \le \frac1{T}\left(\eps(p) + 2U\rho\right)$$
so that
$$\gamma(p) \ge \left(e^{\frac1{T}\left(\eps(p) + 2U\rho\right)}-1\right)^{-1}.$$
Integrating this inequality over $\T^3$ gives us
$$\|\gamma\|_1 \ge J(T,2U\rho).$$
Combining this with \eqref{int_gamma<rho} is a contradiction with the assumption $J(T,2U\rho) \ge \rho = J(T_2,2U\rho)$. It follows that $\rho_0 = 0$. Thus, joint convexity in $(\gamma, \alpha)$ implies uniqueness of the minimizer. This ends the proof of the theorem.
\end{proof}

\noindent {\bf Data availability.} Data sharing is not applicable to this article as no new data were created or analyzed in this study.
\medskip
\\
\noindent \textbf{Acknowledgements.} The work of both authors was supported by the Polish-German NCN-DFG grant Beethoven Classic 3 (project no. 2018/31/G/ST1/01166). We would like to thank the anonymous referees for their valuable comments and suggestions, which helped improve the quality of this paper.

\appendix

\section{Derivation of the functional} \label{app:derivation}
For an even number $L \in 2\mathbb{N}$ we consider a (finite) lattice $\Lambda \subseteq \Z^3$ of the form
$$\Lambda = \left[-\frac{L}{2},\frac{L}{2}\right]^3 \cap \Z^3$$
equipped with the periodic boundary condition. The first Brillouin zone (or the Fourier dual lattice) $\Lambda^*$ is then defined as
$$\Lambda^* = 
\left\{\frac{2\pi}{L+1}j -\pi \, \colon j = 0,1,\dots,L \right\}^3 = \left\{-\pi, -\frac{(L-1)\pi}{L+1}, -\frac{(L-3)\pi}{L+1},\dots, \frac{(L-3)\pi}{L+1}, \frac{(L-1)\pi}{L+1}\right\}^3$$
A single particle on the lattice is described by the element of the one-body Hilbert space 
$$\h = \ell^2(\Lambda).$$
We will denote elements of $\h$ as $\psi = (\psi(x))_{x \in \Lambda}$. When describing the system of $N$ bosonic particles we consider the symmetrized tensor product of the one particle Hilbert spaces
$$\h_N := \bigotimes_{\text{sym}}^N \h,$$
that is this the subspace of $\h^{\otimes N}$ consisting of functions $\psi(x_1,\dots,x_N)$ such that for any permutation $\sigma \in S_N$ we have
$$\psi(x_1,\dots,x_N) = \psi(x_{\sigma(1)},\dots, x_{\sigma(N)}).$$
The Hamiltonian $H_N$ of the interacting particles acting on the space $\h_N$ is given by
$$H_N = -t\sum_{i=1}^N \Delta_i + \frac12\sum_{i,j=1}^N U\delta_{i,j}.$$
Here $\Delta_i$ is the lattice Laplace operator acting on the $i$-th variable:
$$\Delta_i \psi(x_1,\dots,x_N) = \sum_{k=1}^d [\psi(x_1,\dots,x_i + e_k,\dots,x_N) + \psi(x_1,\dots,x_i - e_k,\dots,x_N) - 2\psi(x_1,\dots,x_N)],$$
where $e_k$ is the vector consisting of $1$ on the $k$-th position and zero elsewhere.

The second quantization \cite{DerGer-13} of $H_N$ is then given by 
\begin{equation}
\label{BHHam2}
H =  -t\sum_{ \left\langle x, y \right\rangle \subset \Lambda } a^*_x a_y +6t\sum_{x \in \Lambda}n_x + \frac{U}{2} \sum_{x\in \Lambda} n_x \left(n_x - 1 \right)    
\end{equation}
where $a_x$ and $a_x^*$ are respectively annihilation and creation operators of a particle on lattice side $x$, $n_x$ is a particle number operator $n_x = a_x^*a_x$ and the summation in the first term is taken over the ordered pairs of nearest neighbors $\la x, y\ra$ (ordered means $\la x,y\ra \ne \la y, x\ra$). After rewriting in the momentum space, i.e. 
$$a_p = a(e_p), \qquad \text{with} \qquad e_p(x) = \frac1{\sqrt{|\Lambda|}} e^{i p\cdot x},$$
the Hamiltonian becomes 
\begin{equation}
\label{hamiltonian_p}
H=\sum_p (\eps_p-\mu) a^*_p a_p+\frac{U}{2 |\Lambda|}\sum_{p,q,k}  a_{p+k}^* a_{q-k}^* a_q a_p
\end{equation}
where
$$\eps_p=2t\sum_{j=1}^3(1 - \cos p_j) = 4t\sum_{j=1}^3 \sin^2 \left(\frac{p_j}{2}\right).$$ 
Here we already added the operator $-\mu \sum_x n_x=-\mu \sum_p a_p^* a_p$ corresponding to the grand-canonical setting (with this term the Hamiltonian is sometimes called the grand canonical Hamiltonian). 

To reduce the number of parameters we will rescale the Hamiltonian by a factor $1/t$, that is we divide $H$ by $t$ and replace $U/t \to U$ and $\mu/t \to \mu$ (note that this is equivalent to setting $t=1$). From now on we will consider the rescaled Hamiltonian only.

Now we will recall the notion of quasi-free states. By definition these are the states that satisfy Wick's rule, which, in case that is relevant to us, means
$$\langle a_{p+k}^* a_{q-k}^* a_q a_p\rangle =\langle a^*_{p+k}a^*_{q-k}\rangle\langle a_q a_p\rangle+\langle a^*_{p+k}a_q\rangle\langle a^*_{q-k}a_p\rangle+\langle a^*_{p+k}a_p\rangle\langle a^*_{q-k}a_q\rangle$$
and that the expectations of odd number of operators is zero. Here $\langle \cdot \rangle$ is an expectation value in the chosen state $\omega$ (i.e. positive trace class operator on the Hilbert space with trace equal to one), that is $\langle A \rangle = \Tr A\omega$.

We will further consider the quasi-free states   that are also translation invariant, i.e. they satisfy 
\begin{align*}
   \langle a_{p+k}^* a_{q-k}^* a_q a_p\rangle &= \langle a^*_{p+k}a^*_{q-k}\rangle\langle a_q a_p\rangle \delta_{p+k,-q+k}\delta_{p,-q}+\langle a^*_{p+k}a_q\rangle\langle a^*_{q-k}a_p\rangle \delta_{p+k,q}\delta_{q-k,p}\\&+\langle a^*_{p+k}a_p\rangle\langle a^*_{q-k}a_q\rangle \delta_{p+k,p}\delta_{q-k,q}.
\end{align*}
Now, since we are dealing with bosons, we want to include the possibility of the occurrence of a Bose-Einstein condensate. Bogoliubov \cite{Bogo} suggested to do it by introducing the so-called c-number substitution (this has been then rigorously justified in \cite{LieSeiYng-05}). Mathematically, this can be done by acting with a Weyl transformation on the zero momentum mode operator
$$W^*a_pW = a_p + \delta_{0,p}\sqrt{|\Lambda|\rho_0}, \qquad W^*a_p^*W = a_p^* + \delta_{0,p}\sqrt{|\Lambda|\rho_0}$$
where $\rho_0 \ge 0$ is a parameter that describes the density of the condensate.

The Bogoliubov trial states are now obtained by considering expectation values over quasi-free states with a Weyl transformation of the observable on top of that.

More precisely, in order to evaluate the energy of the system in a Bogoliubov trial state, we first formally replace $a_0 \to a_0 + \sqrt{|\Lambda|\rho_0}$ and $a_0^* \to a_0^* + \sqrt{|\Lambda|\rho_0}$ (the action of the Weyl operator) and further evaluate expectation values of product of creation and annihilation operators  using Wick's rule and translation invariance property. Performing the computation we get
\begin{equation*}
\begin{aligned}
    \langle H \rangle &= \sum_p (\eps_p-\mu) \langle a^*_p a_p\rangle - \mu|\Lambda|\rho_0+\frac{U}{2 |\Lambda|} \sum_{p,q,k} \langle a_{p+k}^* a_{q-k}^* a_q a_p \rangle \\&+\frac{U}{2 |\Lambda|} \cdot |\Lambda|\rho_0 \sum_{p,q,k}\Big[  \langle a_q a_p \rangle\delta_{p+k,0}\delta_{q-k,0} + \langle a^*_{q-k}a_p \rangle \delta_{p+k,0}\delta_{q,0} + \langle a^*_{q-k}a_q \rangle \delta_{p+k,0}\delta_{p,0} \\&+ \langle a^*_{p+k}a_p \rangle \delta_{q-k,0}\delta_{q,0} + \langle a^*_{p+k}a_q \rangle \delta_{q-k,0}\delta_{p,0} + \langle a^*_{p+k}a^*_{q-k} \rangle \delta_{q,0}\delta_{p,0}   \Big] + \frac{U}{2|\Lambda|} \cdot |\Lambda|^2\rho_0^2 \\& = \sum_p (\eps_p-\mu) \langle a^*_p a_p\rangle - \mu|\Lambda|\rho_0 +\frac{U}{2 |\Lambda|}\sum_{p,q}\left[ \langle a_{p}^* a_{-p}^*\rangle \langle a_q a_{-q} \rangle  + 2\langle a_{p}^* a_{p}\rangle \langle a^*_q a_{q} \rangle \right] \\&+ \frac{U}{2}\rho_0 \sum_p\left[ 2\langle a_p a_{-p} \rangle + 4\langle a_p^* a_{p} \rangle \right] + \frac{U}{2}|\Lambda|\rho_0^2 \\& = \sum_p (\eps_p-\mu) \gamma(p) - \mu|\Lambda|\rho_0 +\frac{U}{2 |\Lambda|}\left[\sum_{p,q} \left(\alpha(p)\alpha(q)  + 2\gamma(p)\gamma(q) \right)\right] \\&+ \frac{U}{2}\rho_0 \left[\sum_p \left(2\alpha(p) + 4 \gamma(p)\right)\right] + \frac{U}{2}|\Lambda|\rho_0^2,
\end{aligned}
\end{equation*}
where we have introduced
\begin{align*}
\gamma(p) &= \langle a^*_pa_p\rangle,\\
\alpha(p) &= \langle a_pa_{-p}\rangle
\end{align*}
and additionally assumed $\langle a_pa_{-p}\rangle = \langle a^*_pa^*_{-p}\rangle$, so that $\alpha$ is a real function. The function $\gamma$ is the density of particles outside the condensate and $\alpha$ is a function that describes paring in the system. Those are related to the (generalized) one particle density matrix. The property that this matrix is positive translates to the fact that
$$\alpha^2(p) \le \gamma(p)(\gamma(p)+1).$$
This is the reason we include this condition in the definition of the domain \eqref{domain}.

In order to describe the system at positive temperatures one needs to find the entropy part of the functional. Since the Weyl transformation is unitary, it does not influence the von Neumann entropy. Its value for quasi-free states has been derived in \cite[Appendix A.3.]{NapReuSol1}.    

Finally, evaluating $\langle H -TS \rangle$, dividing this result by $|\Lambda|$ and formally taking the macroscopic limit $|\Lambda| \to \infty$ (i.e. assuming $\frac1{|\Lambda|} \sum_p \to \int_{\T^3} dp$ with normalized measure $dp$) we obtain the free energy density functional
\begin{align*}
\F(\gamma, \alpha, \rho_0) &= \int_{\T^3}(\eps (p) - \mu)\gamma(p) dp - \mu\rho_0 - TS(\gamma, \alpha) \\&+ \frac{U}{2}\left(\int_{\T^3} \alpha(p)dp\right)^2 + U\left(\int_{\T^3}\gamma(p)dp\right)^2 \\&+ U\rho_0\int_{\T^3} \alpha(p) dp + 2U\rho_0\int_{\T^3} \gamma(p) dp + \frac{U}{2}\rho_0^2,
\end{align*}
which is the result \eqref{BBH}.

\end{document}